\documentclass[aps,11pt,twoside, nofootinbib, superscriptaddress]{revtex4}
\usepackage{amsthm}
\usepackage{amsmath,latexsym,amssymb,verbatim,enumerate,graphicx}
\usepackage{color}

\newtheorem{definition}{Definition} 
\newtheorem{prop}[definition]{Proposition}
\newtheorem{lemma}[definition]{Lemma}

\newtheorem{thm}[definition]{Theorem}

\newtheorem{remark}[definition]{Remark}

\newtheorem*{rep@theorem}{\rep@title}
\newcommand{\newreptheorem}[2]{%
\newenvironment{rep#1}[1]{%
 \def\rep@title{#2 \ref{##1} (restatement)}%
 \begin{rep@theorem}}%
 {\end{rep@theorem}}}
\makeatother

\newreptheorem{thm}{Theorem}
\newreptheorem{lem}{Lemma}

\def\ba#1\ea{\begin{align}#1\end{align}}
\def\ban#1\ean{\begin{align*}#1\end{align*}}

\newcommand{\ot}{\otimes}
\newcommand{\be}{\begin{equation}}
\newcommand{\ee}{\end{equation}}

\def\tgr{\textcolor{black}}
\def\good{\text{Good}}
\def\Extr{\text{Extr}}

\def\dcintro{d_{\text{comp}}}
\def\dcours{\tilde d_{\text{comp}}}
\def\dnow{d}
\def\dwithw{d'}
\def\dcintroII{d_{\text{comp}}^{\,\text{II}}}
\def\dcoursII{\tilde d_{\text{comp}}^{\,\text{II}}}
\def\dnowII{d_{\,\text{II}}}
\def\dwithwII{d'_{\,\text{II}}}
\def\ein{\text{\tt in}_E}
\def\eout{\text{\tt out}_E}
\def\sv{\text{\tt sv}}
\def\abin{\text{\tt in}}
\def\about{\text{\tt out}}
\def\sve{\text{\tt sv}_E}
\def\svin{\text{\tt sv}_1}
\def\svhash{\text{\tt sv}_2}

\def\benum{\begin{enumerate}}
\def\eenum{\end{enumerate}}

\def\squareforqed{\hbox{\rlap{$\sqcap$}$\sqcup$}}
\def\qed{\ifmmode\squareforqed\else{\unskip\nobreak\hfil
\penalty50\hskip1em\null\nobreak\hfil\squareforqed
\parfillskip=0pt\finalhyphendemerits=0\endgraf}\fi}
\def\endenv{\ifmmode\;\else{\unskip\nobreak\hfil
\penalty50\hskip1em\null\nobreak\hfil\;
\parfillskip=0pt\finalhyphendemerits=0\endgraf}\fi}

\newcommand{\<}{\langle}
\renewcommand{\>}{\rangle}

\def\be{\begin{equation}}
\def\ee{\end{equation}}
\def\ben{\begin{eqnarray}}
\def\een{\end{eqnarray}}
\def\ot{\otimes}

\def\bei{\begin{itemize}}
\def\eei{\end{itemize}}

\def\E{{\mathbb{E}}}

\def\ep{\epsilon}

\mathchardef\ordinarycolon\mathcode`\:
\mathcode`\:=\string"8000
\def\vcentcolon{\mathrel{\mathop\ordinarycolon}}
\begingroup \catcode`\:=\active
  \lowercase{\endgroup
  \let :\vcentcolon
  }

\newcommand{\nc}{\newcommand}
 \nc{\proj}[1]{|#1\rangle\!\langle #1 |} 
\nc{\avg}[1]{\langle#1\rangle}

\nc{\conv}{\operatorname{conv}}
\nc{\smfrac}[2]{\mbox{$\frac{#1}{#2}$}} \nc{\Tr}{\operatorname{Tr}}
\nc{\ox}{\otimes} \nc{\dg}{\dagger} \nc{\dn}{\downarrow}
\nc{\lmax}{\lambda_{\text{max}}}
\nc{\lmin}{\lambda_{\text{min}}}

\nc{\csupp}{{\operatorname{csupp}}}
\nc{\qsupp}{{\operatorname{qsupp}}} \nc{\var}{\operatorname{var}}
\nc{\rar}{\rightarrow} \nc{\lrar}{\longrightarrow}
\nc{\poly}{\operatorname{poly}}
\nc{\polylog}{\operatorname{polylog}} \nc{\Lip}{\operatorname{Lip}}
\nc{\Om}{\Omega}
\nc{\wt}[1]{\widetilde{#1}}

\def\>{\rangle}
\def\<{\langle}

\def\EE{\mathbb{E}}

\def\bu{{\textbf{u}}}
\def\bx{\textbf{x}}
\def\pxu{P_{{x}_{< l},{u}_{< l}}({x}_l | {u}_l)}
\def\pxutot {P(x_1,\ldots,x_n|u_1,\ldots,u_n)}
\def\xuseq{x_1,u_1,\ldots,x_n,u_n}

\nc{\glneq}{{\raisebox{0.6ex}{$>$}  \hspace*{-1.8ex} \raisebox{-0.6ex}{$<$}}}
\nc{\gleq}{{\raisebox{0.6ex}{$\geq$}\hspace*{-1.8ex} \raisebox{-0.6ex}{$\leq$}}}

\nc{\vholder}[1]{\rule{0pt}{#1}}
\nc{\wh}[1]{\widehat{#1}}
\nc{\h}[1]{\widehat{#1}}

\nc{\ob}[1]{#1}

\def\beq{\begin {equation}}
\def\eeq{\end {equation}}

\def\be{\begin{equation}}
\def\ee{\end{equation}}

\def\setadefnc{S^{\xi}}

\def\Ddef{D_{\epdef}}
\def\epazdef{\ep_{Az}}

\def\epdef{\ep_{deF}}
\def\ACC{\text{ACC}}
\def\ACCdef{\text{ACC}}

\def\ACCd{\text{ACC}}

\def\Lazuma{L}
\def\Lazumadef{\overline{L}}

\def\Xgood{X^{(z,u,e)}_{\text{good}}}
\def\Xbad{X^{(z,u,e)}_{\text{bad}}}

\def\Azumad{A^{\delta_e}}

\def\Azumadnc{A^{\delta_{Az}}}

\def\Adef{A^{\delta_{Az}}}
\def\Adefo{A^{\delta_{Az},1}}
\def\Adefj{A^{\delta_{Az},j}}

\nc{\eq}[1]{(\ref{eq:#1})} 
\nc{\eqs}[2]{\eq{#1} and \eq{#2}}

\nc{\eqn}[1]{Eq.~(\ref{eqn:#1})}
\nc{\eqns}[2]{Eqs.~(\ref{eqn:#1}) and (\ref{eqn:#2})}

\nc{\region}{\cS\cW}

\newenvironment{protocol*}[1]
  {
    \begin{center}
      \hrulefill\\
      \textbf{#1}
  }
  {
    \vspace{-1\baselineskip}
    \hrulefill
    \end{center}
  }

\begin{document}


\title{Robust Device-Independent Randomness Amplification with Few Devices}

\author{Fernando G.S.L. Brand\~{a}o}
\affiliation{Quantum Architectures and Computations Group, Microsoft Research, Redmond, WA (USA)}
\affiliation{Department of Computer Science, University College London}

\author{Ravishankar Ramanathan}
\affiliation{National Quantum Information Center of Gda\'{n}sk, 81-824 Sopot, Poland}
\affiliation{Institute of Theoretical Physics and Astrophysics, University of Gda\'{n}sk, 80-952 Gda\'{n}sk, Poland}

\author{Andrzej Grudka}
\affiliation{Faculty of Physics, Adam Mickiewicz University, 61-614 Pozna\'{n}, Poland}

\author{Karol Horodecki}
\affiliation{National Quantum Information Center of Gda\'{n}sk, 81-824 Sopot, Poland}
\affiliation{Institute of Informatics, University of Gda\'{n}sk, 80-952 Gda\'{n}sk, Poland}

\author{Micha{\l} Horodecki}
\affiliation{National Quantum Information Center of Gda\'{n}sk, 81-824 Sopot, Poland}
\affiliation{Institute of Theoretical Physics and Astrophysics, University of Gda\'{n}sk, 80-952 Gda\'{n}sk, Poland}

\author{Pawe{\l} Horodecki}
\affiliation{National Quantum Information Center of Gda\'{n}sk, 81-824 Sopot, Poland}
\affiliation{Faculty of Applied Physics and Mathematics, Technical University of Gda\'{n}sk, 80-233 Gda\'{n}sk, Poland}

\author{Tomasz Szarek}
\affiliation{Institute of Mathematics, University of Gda\'{n}sk, 80-952 Gda\'{n}sk, Poland}

\author{Hanna Wojew\'{o}dka}
\affiliation{National Quantum Information Center of Gda\'{n}sk, 81-824 Sopot, Poland}
\affiliation{Institute of Theoretical Physics and Astrophysics, University of Gda\'{n}sk, 80-952 Gda\'{n}sk, Poland}
\affiliation{Institute of Mathematics, University of Gda\'{n}sk, 80-952 Gda\'{n}sk, Poland}




\date{\today}

\begin{abstract}

Randomness amplification is the task of transforming a source of somewhat random bits into a source of fully random bits. Although it is impossible to amplify randomness from a single source by classical means, the situation is different considering non-local correlations allowed by quantum mechanics. Here we give the first device-independent protocol for randomness amplification using a constant number of devices. The protocol involves four devices, can amplify any non-deterministic source into a fully random source, tolerates a constant rate of error, and has its correctness based solely on the assumption of no-signaling between the devices. In contrast all previous protocols either required an unbounded number of devices, or could only amplify sources sufficiently close to fully random. 

\end{abstract}

\maketitle

\section{Introduction}

Randomness is a useful resource in a variety of applications, ranging from numerical simulations to cryptography. However, almost always one needs ideal or close to ideal randomness, while typical random sources are far from ideal. Randomness amplification, defined as a process that maps a source of randomness into another closer to an ideal source, is a potential solution to this problem. But is randomness amplification possible at all? 

We model a source of randomness as an $\varepsilon$-Santha-Vazirani source ($\varepsilon$-SV source), given by a probability distribution $p(\alpha_1, \ldots, \alpha_n)$ over bit strings such that, for every $i \leq n$,
\begin{equation} \label{SVdef}
\frac{1}{2} - \varepsilon \leq p(\alpha_i | \alpha_1, \ldots, \alpha_{i-1}) \leq \frac{1}{2} + \varepsilon.
\end{equation}
The previous equation is the only assumption on the source, which otherwise can be arbitrary. Given an $\varepsilon$-SV source, is it possible to process the bits so that the quality of the randomness is improved? In particular, can we obtain a fully random bit by processing an arbitrary large number of bits from any $\varepsilon$-SV source? Santha and Vazirani answered the question in the negative \cite{SV}: It is impossible to improve the quality of the randomness of SV sources \footnote{Similar results are also known for other different sources of randomness \cite{CG}.}. However, their argument only applies to classical protocols and it leaves open the possibility that the situation might be different once quantum resources are considered. 

Indeed, it is trivial to generate randomness in quantum mechanics, e.g. by preparing a pure state and measuring it in a complementary basis. However, this assumes that one has full control of the state preparation and the measurement. A more demanding task is to generate randomness in a device independent manner, treating the quantum system as a black-box and obtaining randomness only as a consequence of the correlations in measurement outcomes and fundamental physical principles such as the no-signaling principle. The existence of non-local quantum correlations violating Bell inequalities already suggests that device-independent randomness amplification could be achieved. However, the violation of Bell inequalities requires that the measurements be performed in a random manner, independent of the system upon which they are performed \cite{Barrett, Hall}. 

In a seminal work \cite{Renner}, Colbeck and Renner showed that despite this difficulty non-local quantum correlations can be used to amplify the randomness of Santha-Vazirani sources that are sufficiently close to fully random. This result was later improved by Gallego \textit{et al.} \cite{Acin}, who gave a~protocol using quantum non-local correlations to amplify general SV sources, as long as they are not deterministic. Neither of the two protocols tolerate noise however. In \cite{our}, we gave a different protocol that is robust to noise and can transform any non-deterministic SV source into a fully random one. A major drawback of the protocols in \cite{Acin, our} is that they require an infinite number of space-like separated devices. Therefore a natural open question is the existence of a randomness amplification protocol using a fixed number of devices, on the one hand, and allowing for the amplification of arbitrary non-deterministic SV sources, on the other hand. See also \cite{Grudka, Pawlowski, Alastair, Scarani, Plesch, Augusiak, CY13, CSW14} for more recent work in the area. 

A related, but distinct, task to randomness amplification is device-independent randomness expansion, where one assumes that a seed of perfect random bits is available and the goal is to expand it into a larger random bit string. Quantum non-locality has found application also in this latter task \cite{Colbeck10, Pironio, Colbeck, Acin2, Fehr11, Pironio13, VV12a, Coudron}, as well as in device-independent quantum key-distribution (see e.g. \cite{BHK, Masanes09, Hanggi, VV12b}).

 \subsection{Results}

In this paper we overcome the shortcomings of previous protocols and obtain:

\begin{thm} [informal]  \label{mainthm}
For every $\varepsilon < \frac{1}{2}$, there is a protocol using an $\varepsilon$-SV source and four no-signaling devices with the following properties:

\begin{itemize}
\item Using the devices $\poly(n, \log(1/\delta))$ times, the protocol either aborts or produces $n$ bits which are $\delta$-close to uniform and independent of any side information (e.g. held by an adversary).

\item Local measurements on many copies of a four-partite entangled state, with $\poly(1 - 2\varepsilon)$ error rate, give rise to devices that do not abort the protocol with probability larger than $1 - 2^{- \Omega(n)}$.

\end{itemize}
\end{thm}

An important assumption in the theorem above is that the SV source is independent of the joint device shared by the honest parties and adversary (see Sec. \ref{sec:assumptions} for details). The theorem is based on Protocol I, given in Fig. \ref{protocolsingle} , using the four-partite Bell inequality given in Section \ref{subsec:Bell}. See Theorem \ref{thm:main_protocol_I} in Section \ref{sec:protocol1} for a precise formulation of the theorem.

Theorem \ref{mainthm} has the drawback that the extractor used in Protocol \ref{protocolsingle} is non-explicit (we only know that it exists by the probabilistic method). We can improve on this aspect if we are willing to increase the number of no-signaling devices and worsen the efficiency of the protocol with the output error.

\begin{thm}[informal]   \label{mainthm2}
For every $\varepsilon < \frac{1}{2}$, there is a protocol using an $\varepsilon$-SV source and eight no-signaling devices with the following properties:

\begin{itemize}
\item Using the devices $\poly(n, 1/\delta)$ times, the protocol either aborts or produces $n$ bits which are $\delta$-close to uniform and independent of any side information (e.g. held by an adversary).

\item Local measurements on many copies of a four-partite entangled state, with $\poly(1 - 2\varepsilon)$ error rate, give rise to devices that do not abort the protocol with probability larger than $1 - 2^{- \Omega(n)}$.
\end{itemize}

The protocol is fully explicit and runs in $\poly(n, 1/\delta)$ time.

\end{thm}

The protocol is given in Fig. \ref{protocoltwo}. Its proof of correctness is analogous to the proof of Theorem \ref{mainthm}, with a new ingredient that we show how to simulate two independent sources selecting subsystems at random 
(applying an analogue of the de Finetti Theorem of \cite{Brandao} to subsystems selected by a SV source). The advantage of this new step is that we have three independent sources and so can use known explicit and computationally efficient extractors.

\subsection{Features of the Protocols and Comparison with Previous Works}

In this section we compare our result with similar work in the area. The main protocols for randomness amplification proposed so far are summarized in Table \ref{table-properties}. There are several aspects of a protocol for randomness amplification to be considered:

\begin{itemize}

\item \textbf{Source-Device Correlations:} In this paper (as in \cite{Renner, Acin}) we assume that the device shared by the honest parties and Eve is independent of the imperfect randomness source\footnote{The analogous classical problem would be that one wishes to extract fully random bits from a weak source of randomness and an unknown channel. While this is clearly impossible, the situation is different when one considers non-local correlations.}. A less demanding requirement is to require only that the source has randomness \textit{conditioned} on the devices \footnote{This is the assumption considered in the work of Chung, Shi, and Wu \cite{CY13} discussed below.}. 

\item \textbf{Number of Devices:} This is the total number of different no-signalling devices that are used in the protocol. For a protocol to be practical, it must involve only a small constant number of devices. 

\item \textbf{Robustness:} This is the amount of error per basic element that the protocol can withstand while still working correctly. A basic element is either a two-qubit gate, a qubit measurement, or the storage of one qubit for one time step. For the protocol to be practical it must tolerate a constant amount of noise per basic element. 

\item \textbf{Eavesdropper:} This represents the type of adversary under which the protocol is secure. It can be either a quantum adversary or a more powerful no-signalling adversary.  

\item \textbf{Composability:} This represents whether the protocol is composable (being secure even if Eve measures her part of the device after learning some of the output random bits). 

\item \textbf{Source:} Two types of sources have been considered so far: Santha-Vazirani sources and the less demanding min-entropy sources (where one only requires that the source has a certain amount of min-entropy). 

\item \textbf{Public Source:} This item makes a distinction between the way the bits of the random source are distributed. In a public source the bits are drawn from the source and then communicated to all the parties (including Eve). In a private source, the honest parties have exclusive access to a part of the source, and the eavesdropper only learn about the bits in this part from her correlations with the other parts of the source under the adversary's control. 


\item \textbf{Run Time:} This item quantifies the computational complexity of the protocol. 
\end{itemize}

\begin{table}
\begin{tabular}{@{}c||c|c|c|c|c|c|c|c|c|c|c@{}}
     		&    \vtop{\hbox{\strut Source-Device}\hbox{\strut Correlations}}   &  \# Devices 	&  \small{Robustness}	&  \small{Eve} 	& \small{Comp.} &  \small{Source}	  &   \vtop{\hbox{\strut Public}\hbox{\strut Source}}   & \small{Run time} \\
\hline \hline
\vtop{\hbox{\strut Colbeck-}\hbox{\strut Renner \cite{Renner}}}      	&     \textcolor{black}{indep.}     		&  \textcolor{black}{2}    	&     $  \textcolor{black}{\frac{1}{\text{poly}(m)}}$    	&      \textcolor{black}{NS}    	&    \textcolor{black}{no}      	&  \vtop{\hbox{\strut \hspace{0.3 cm}  \textcolor{black}{SV}}\hbox{\strut \textcolor{black}{$\varepsilon < .08$}}}		&      \textcolor{black}{yes}   		 &     \textcolor{black}{poly$(\frac{1}{\delta},m)$}   	 \\
\hline
Gallego \textit{et al} \cite{Acin}  		&    \textcolor{black}{indep.}    	&  \textcolor{black}{poly($\frac{1}{\delta}, m$) }  	&  $  \textcolor{black}{\frac{1}{\text{poly}(m)}}$    &     \textcolor{black}{NS} 	     &  \textcolor{black}{yes}    &      \textcolor{black}{SV}	     &    \textcolor{black}{yes}  &    \textcolor{black}{$2^{\text{poly}(\log\frac{1}{\delta}, m)}$}  	 \\

\hline
\vtop{\hbox{\strut  \textbf{This Paper}}\hbox{\strut Protocol I}}         	 &      \textcolor{black}{indep.}              		&    \textcolor{black}{4}    	&   \textcolor{black}{ $\Omega(1 - 2\varepsilon)$   }    	&     \textcolor{black}{NS}  		&       \textcolor{black}{yes}      	&        \textcolor{black}{SV}	          &         \textcolor{black}{no}       &    \textcolor{black}{$2^{\text{poly}(\log\frac{1}{\delta}, m)}$} 	 \\

\hline
\vtop{\hbox{\strut  \textbf{This Paper}}\hbox{\strut Protocol II}}    &     \textcolor{black}{indep.}            		&        \textcolor{black}{8}   	&     \textcolor{black}{$\Omega(1 - 2\varepsilon)$}       	&           \textcolor{black}{NS}   	&         \textcolor{black}{yes}        &      \textcolor{black}{SV}     &      \textcolor{black}{no}    	&         \textcolor{black}{$\poly(\frac{1}{\delta}, m)$}	 \\

\hline
Chung-Shi-Wu \cite{CSW14}    	 &  \vtop{\hbox{\strut \hspace{0.2 cm}  \textcolor{black}{positive cond.}}\hbox{\strut  \textcolor{black}{min-entropy}}}  &    \textcolor{black}{poly($\frac{1}{\delta}$)}    &     \textcolor{black}{$\Omega(1)$}      	&   \textcolor{black}{Q}	 &       \textcolor{black}{yes}     	&          \textcolor{black}{$H_{\text{min}}$}       &        \textcolor{black}{yes}          &    \textcolor{black}{poly($\log \frac{1}{\delta}, m$)} 	\\

\hline
\vtop{\hbox{\strut Mironowicz-}\hbox{\strut Pawlowski  \textcolor{black}{*} \cite{Pawlowski}}}      &   \textcolor{black}{indep.}          		&   \textcolor{black}{2}      	&    $\Omega(1)$  	&  \textcolor{black}{Q}  		&       \textcolor{black}{no}    	&      \textcolor{black}{SV}         &       \textcolor{black}{no}        	&     \textcolor{black}{$\poly(\log \frac{1}{\delta}, m)$} 	 \\

\hline
Bouda \textit{et al} \textcolor{black}{*}    \cite{BPPP14} 		 &     \vtop{\hbox{\strut \hspace{0.2 cm}  \textcolor{black}{positive}}\hbox{\strut  \textcolor{black}{min-entropy}}} 		&     \textcolor{black}{3}   	&      \textcolor{black}{$\frac{1}{\text{poly}(m)}$}      	&    \textcolor{black}{Q}	 &       \textcolor{black}{no}   	&      \textcolor{black}{$H_{\text{min}}$}        &      \textcolor{black}{yes}      &  
      \textcolor{black}{$\poly(\log \frac{1}{\delta}, m)$} 	\\	
\hline
\end{tabular}
\vspace{0.3cm} \caption{Comparison of protocols for randomness amplification. In the table we used: 1) $m$ : number of output bits;
2) $\delta$: distance from uniform of output bits; 3) NS and Q stands for no-signaling and quantum adversaries, respectively. black indicates less demanding assumptions/better parameters, while red indicates more depanding assumptions/worse parameters. \textcolor{black}{*}The protocols proposed in \cite{Pawlowski} and \cite{BPPP14} were not given a full security proof so far.} \label{table-properties}
\end{table}

In addition to the two protocols analysed in this paper, several other interesting protocols with complementary advantages have been proposed. Apart from the work of Colbeck and Renner \cite{Renner} and Gallego \textit{et al} \cite{Acin} discussed in the introduction, another interesting protocol was proposed by Chung, Shi, and Wu \cite{CSW14} (after the first version of this paper appeared). Their protocol is sound against  quantum adversaries and can amplify the randomness of every min-entropy source. Moreover it has the distinguishing feature that the correct functioning of the protocol is guaranteed based only on the source having positive min-entropy conditioned on the (quantum) state of the devices. A~drawback of the protocol is that it requires an unbounded number of devices.

Another interesting development was the work of Coudron and Yuen \cite{CY13}, who shown that any composable protocol can be made to have infinity rate (at the cost of increasing the number of devices by two) and decreasing the robustness. Combining their result with ours one obtains a protocol whose rate is infinite (the number of bits from the SV source only determines the output error).


\subsection{Protocols and Outline of their Proofs of Correctness}

\begin{figure}[h]
\begin{protocol*}{Protocol I}
\begin{enumerate}
\item The $\varepsilon$-SV source is used to choose the measurement settings $u^1_{\leq n}$ for the single device.
The device produces output bits $x^1_{\leq n}$. 
\item The parties perform an estimation of the violation of the Bell inequality in the  device by computing the empirical average $\textit{L}_n :=  \frac{1}{n} \sum_{i=1}^{n} B(x_i, u_i)$. The protocol is aborted unless $\textit{L}_n \leq \delta$ (for $\delta>0$ is a~constant depending only on $\epsilon$).
\item Conditioned on not aborting in the previous step, the parties apply the extractor from part (i) of Lemma \ref{extractors} to the sequence of outputs from the device and further $n$ bits from the $\varepsilon$-SV source.
\end{enumerate}
\end{protocol*}
\caption{Protocol for device-independent randomness amplification from a single (four-partite) device with non-explicit extractor.}
\label{protocolsingle}
\end{figure}

\begin{figure}
\scalebox{0.60}{\includegraphics{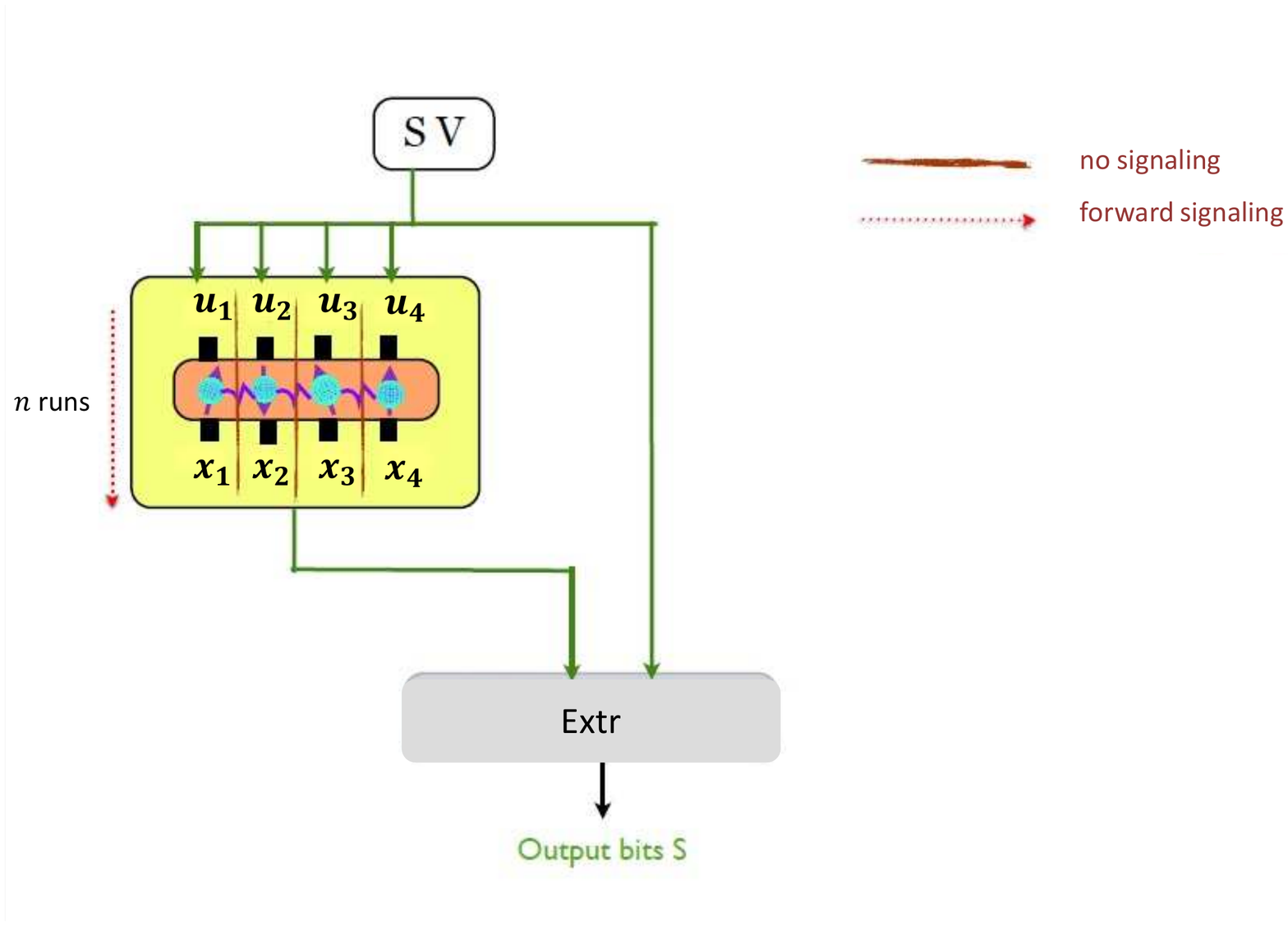}}
\caption{Illustration of the protocol for randomness amplification from a single device (of four no-signaling parts) of $n$ runs by using a non-explicit randomness extractor.}
\label{protocol-fig-1}
\end{figure}

\textit{Protocol I} is simpler than the previous protocols \cite{Renner, Acin, our}. It is given precisely in Fig. \ref{protocolsingle} and illustrated in Fig. \ref{protocol-fig-1}, but its rough structure is the following: First, one uses several bits from the Santha-Vazirani source in order to choose inputs for four no-signaling boxes (each of which is reused many times). Then, one collects the outputs of the boxes and using the empirical data on inputs and outputs decides whether to abort. Then if the protocol is not aborted, one applies a randomness extractor (see Sec. \ref{subsec:rand-ext}) to all output bits of the four devices and to a further set of bits taken from the SV source. The output bits of the protocol, which by Theorem \ref{mainthm} are close to fully random, are just the output bits of the extractor. 

The proof of correctness of the protocol consists of three main steps.

\bei 
\item[(i)] We show that by increasing the output error of the protocol (and the assumption that the source is private) it suffices to prove the security of the protocol considering an adversary that only has classical side-information about the devices of the honest parties. This is achieved in Section \ref{sec:assumptions}.

\item[(i)] We show that, conditioned on passing a certain test on the input and output of the four no-signaling boxes, with high probability over the given input, the output is a source of linear min-entropy. This is achieved in Proposition \ref{prop:max_x_small}. 
\item[(iii)]  We use know results on extractors for two independent sources. 
\eei 

Step (ii) is established by a sequence of implications (which have a similar flavor to the estimation of \cite{Pironio, Pironio13, Fehr11} for the related task of randomness expansion) as follows:
\bei
\item In Sec. \ref{subsec:Azuma} we present an estimation procedure which ensures that with high probability the value of the Bell expression with settings chosen from a SV source is small \footnote{The Bell inequality we consider is such that the maximum possible violation corresponds to the zero value. Thus the smaller the Bell value, the larger the violation.} for a~linear fraction of boxes conditioned on previous inputs and outputs. This will follow from an simple application of Azuma's inequality.
\item In Sec. \ref{subsec:SV-Bell} we show that a small  value of the Bell value with settings chosen from a~SV source implies that for any setting the probability of any output is bounded away from one. This is achieved by solving a linear-program, analogously to the approach of \cite{Acin,our}. 
\item In Sec. \ref{subsec:min-entropy-source}, in turn, we show that if a constant fraction of conditional boxes has probability of outputs bounded away from one,  then, given the input, the output has linear min-entropy.
\eei 

Step (iii) is an application of known results regarding extracting randomness from two independent min-entropy sources 
(see Sec. \ref{subsec:rand-ext}). They say, that one can extract randomness from two or more independent min-entropy sources.
In our case the independence follows from our main assumption saying that device is independent of the SV source.

\begin{figure}[h]
\begin{protocol*}{Protocol II}
\begin{enumerate}
\item The $\varepsilon$-SV source is used to choose the measurement settings $u^1_{\leq M_1}, u^2_{\leq M_2}$ for the $2$ devices.
The devices produce output bits $x^1_{\leq M_1}, x^2_{\leq M_2}$. 


\item 
The measurements in device $j$ (for $j = 1, 2$) are partitioned into $N_j$ blocks of boxes each containing $n$ boxes (so that $M_j = n N_j$).

\item The parties choose at random one block of boxes of size $n$ from each device, using bits from the $\varepsilon$-SV source.
 
\item The parties perform an estimation of the violation of the Bell inequality in the chosen block from each of the two devices by computing the empirical average $\textit{L}^j_n :=  \frac{1}{n} \sum_{i=1}^{n} B(x^j_i, u^j_i)$. The protocol is aborted unless for both of them, $\textit{L}^j_n \leq \delta$.
\item Conditioned on not aborting in the previous step, the parties apply the extractor from part (ii) of Lemma \ref{extractors} to the sequence of outputs from the chosen block in each device.
\end{enumerate}
\end{protocol*}
\caption{Protocol for device-independent randomness amplification from two (four-partite) devices with an explicit, and efficient, extractor.}
\label{protocoltwo}
\end{figure}

It is instructive to point out already here where the non-local nature of our protocol plays a role. It is in the second step that the non-local nature of correlations is exploited. There, looking at the inputs and outputs of the devices obtained, one can \textit{verify} in a device-independent manner that the outputs must have been somewhat random. This is impossible classically without making further assumptions.

\vspace{0.5 cm}

\textit{Protocol II}: In the scenario where we require an explicit extractor (Protocol II described in Fig. \ref{protocoltwo} and illustrated in Fig. \ref{protocol-fig}), a further set of four no-signaling boxes is taken into account. As before, bits from the SV source are input into these boxes and a test is performed on the empirical data of inputs and outputs. In this case, the randomness extractor is applied to the output bits of all eight no-signaling devices and a further set of bits from the SV source.

\begin{figure}
\scalebox{0.60}{\includegraphics{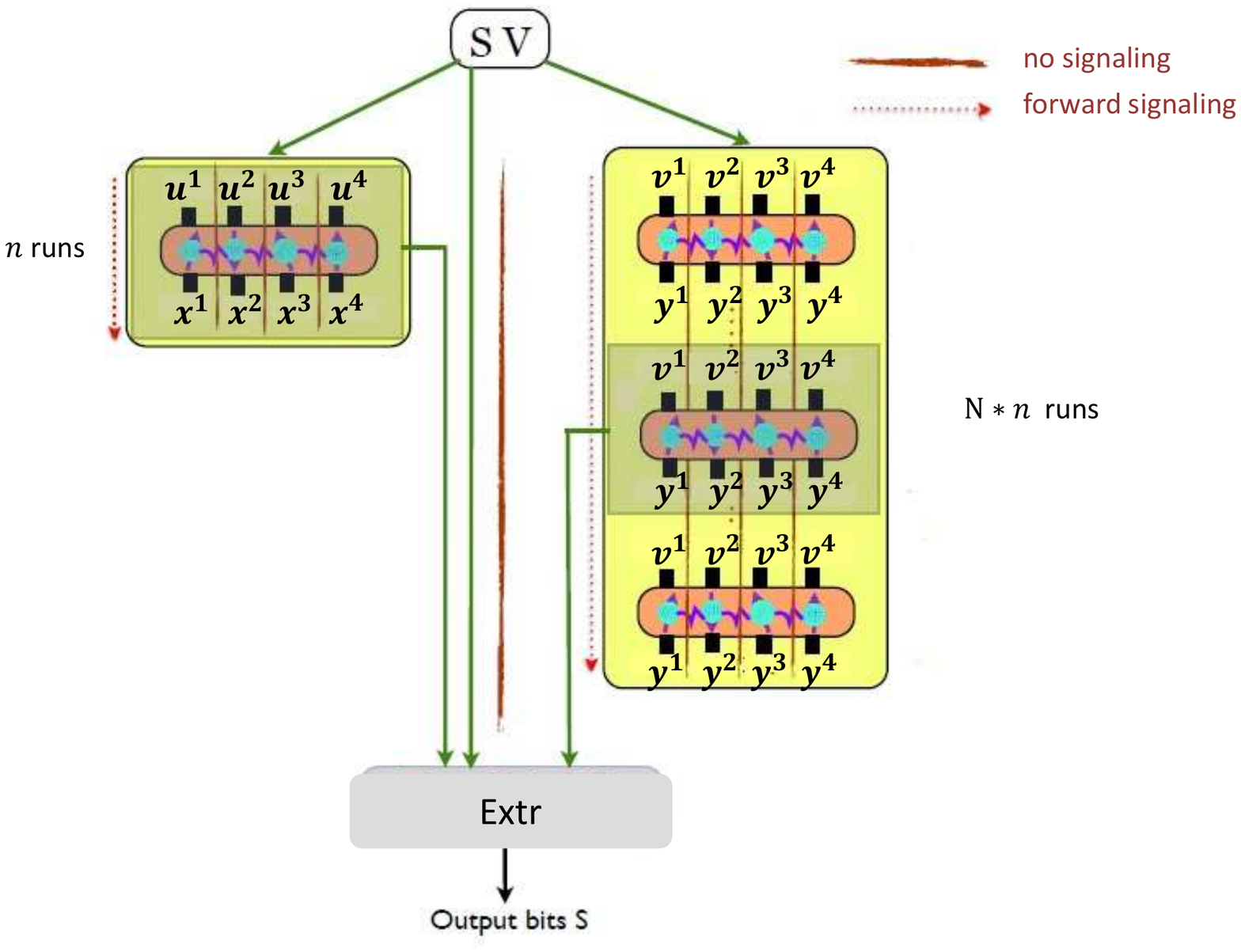}}
\caption{Illustration of the protocol for randomness amplification from two devices (of four no-signaling parts each) with $N_1 = 1$ block of $n$ runs from the first device and $N_2$ blocks of $n$ runs from the second.}
\label{protocol-fig}
\end{figure}

\tgr{The proof of security of Protocol II is contained in Sec. \ref{sec:protocol2}, with the main result stated in Theorem 
\ref{thm:main_protocol_II}.} 
We need here an extra step to reduce the problem to having three independent sources of week randomness. Namely we show that we can effectively work with two independent \tgr{devices 
(of four no-signaling parts each);} by independent we mean devices that, given any fixed inputs, produce uncorrelated outputs.

The idea is to adapt recent de Finetti theorems for no-signaling devices \cite{Brandao, Brandao2} (based on information-theoretical methods) to the situation in which subsystems are selected from a Santha-Vazirani source, instead of being selected uniformly at random. 
In Lemma \ref{lem:deFinetti-bound1} we show that given two no-signaling devices, if we select at random (using a SV source) a block of 
uses from the second device among a sufficiently large number of blocks, this block of will be approximately uncorrelated with the first device. 


\section{Preliminaries}


\subsection{Randomness Extractors}
\label{subsec:rand-ext}
Recall the definition of a min-entropy source:
\begin{definition}[\tgr{The min-entropy source}]
The min-entropy of a random variable $S$ is given by 
\begin{equation}
H_{min}(S) = \min_{s \in supp(S)} \log \frac{1}{P(S = s)}.
\end{equation}
When $S \in \{0,1\}^n$, it is called an $(n, H_{min}(S))$.
\end{definition} 
From multiple weak sources of randomness, one can use independent source extractors \cite{Xin-Li} to extract nearly uniform random bits. 
\begin{definition}[Independent-source extractor]
An independent source extractor is a function $Ext: (\{0,1\}^n)^k \rightarrow \{0,1\}^m$ that acting on $k$ independent $(n, H_{min}(S))$ sources, outputs $m$ bits which are $\xi$-close to uniform, i.e. for $k$ independent $(n, H_{min}(S))$ sources $S_1, \ldots, S_k$ we have
\begin{equation}
\Vert Ext(S_1, \ldots, S_k) - U_m \Vert_1 \leq \xi,
\end{equation}
where $\| . \|_1$ is the variational distance between the two distributions and $U_m$ denotes the uniform distribution on the $m$ bits. 
\end{definition}   

The results about extractors that we will use are summarized in the following lemma:
\begin{lemma}[Extractors Constructions] \label{extractors}

\mbox{}

\bei 
\item[(i)] \cite{CG} There exists a (non-explicit) deterministic extractor that, given two independent sources of min-entropy larger than $h$, outputs $\Omega(h)$ bits $2^{-\Omega(h)}$-close to uniform.

\item[(ii)] \cite{Rao}  There exists an explicit extractor that given three independent sources, one having min-entropy larger than $\tau n$ (for any $\tau > 0$) and the other two larger than $h \geq \log^c(n)$ (with $c > 0$ a universal constant), outputs $\Omega(h)$ bits $2^{- h^{\Omega(1)}}$-close to uniform.The extractor can be implemented in time $\poly(n, h)$. 
\eei 

\end{lemma}

When we say the first extractor is non-explicit we mean that its existence is only guaranteed by the probabilistic method.



Theorem \ref{mainthm} uses the non-explicit extractor for two sources stated in part (i), while Theorem \ref{mainthm2} uses Rao's extractor \cite{Rao} stated in part (ii).

\subsection{The Bell inequality} \label{subsec:Bell}
The inequality we consider involves four spatially separated parties with measurement settings $\textbf{u} = \{\textbf{u}^1, \textbf{u}^2, \textbf{u}^3, \textbf{u}^4\}$ and respective outcomes $\textbf{x} = \{\textbf{x}^1, \textbf{x}^2, \textbf{x}^3, \textbf{x}^4\}$. Each party chooses one of two measurement settings with two outcomes each so that $\textbf{u}^i \in \{0, 1\}$ and $\textbf{x}^i \in \{0,1\}$ for $i \in \{1,..,4\}$. 
The measurement settings for which non-trivial constraints are imposed by the inequality can be divided into two sets
\begin{equation}
\textsl{U}_0 = \{ \{0001\}, \{0010\}, \{0100\}, \{1000\} \} \hspace{0.3 cm} \text{and} \hspace{0.3 cm} \textsl{U}_1 = \{ \{0111\}, \{1011\}, \{1101\}, \{1110\} \}. 
\end{equation}
The inequality is then \cite{Guhne}
\begin{eqnarray}
\label{Bell-ineq}
\sum_{\textbf{x}, \textbf{u}} (\texttt{I}_{\oplus_{i=1}^{4} \textbf{x}^i = 0} \; \texttt{I}_{\textbf{u} \in \textsl{U}_0} \; + \texttt{I}_{\oplus_{i=1}^{4} \textbf{x}^i = 1} \; \texttt{I}_{\textbf{u} \in \textsl{U}_1}) \; P(\textbf{x}|\textbf{u}) \geq 2,
\end{eqnarray}
where the indicator function $\texttt{I}_{L} = 1$ if $L$ is true and $0$ otherwise. The local hidden variable bound is $2$ and there exist no-signaling distributions that reach the algebraic limit of $0$. For any no-signaling box represented by a vector of probabilities $\{P(\textbf{x}|\textbf{u})\}$, the Bell inequality may be written as
\begin{equation}
\textbf{B}.\{P(\textbf{x} | \textbf{u})\} = \sum_{\textbf{x}, \textbf{u}} B(\textbf{x} , \textbf{u}) P(\textbf{x} | \textbf{u}) \geq 2, 
\end{equation}
where $\textbf{B}$ is an indicator vector for the Bell inequality with $2^4 \times 2^4$ entries 
\begin{equation}
B(\textbf{x} , \textbf{u}) = \texttt{I}_{\oplus_{i=1}^{4} \textbf{x}^i = 0} \; \texttt{I}_{\textbf{u} \in \textsl{U}_0} + \texttt{I}_{\oplus_{i=1}^{4} \textbf{x}^i = 1} \; \texttt{I}_{\textbf{u} \in \textsl{U}_1}.
\label{eq:Bell_indicator}
\end{equation} 
Consider the quantum state 
\begin{equation} \label{state}
|\Psi \rangle = \frac{1}{\sqrt{2}} (|\phi_{-}\rangle |\tilde{\phi}_{+} \rangle + |\psi_{+} \rangle |\tilde{\psi}_{-} \rangle),
\end{equation}
where $|\phi_{-} \rangle = \frac{1}{\sqrt{2}}(|0\rangle |0\rangle - |1\rangle |1\rangle)$, $|\psi_{+}\rangle = \frac{1}{\sqrt{2}} (|0\rangle |1 \rangle + |1\rangle |0\rangle)$, $|\tilde{\phi}_{+}\rangle = \frac{1}{\sqrt{2}}(|0\rangle |+\rangle + |1\rangle |-\rangle)$, and $|\tilde{\psi}_{-}\rangle = \frac{1}{\sqrt{2}}(|0\rangle |-\rangle - |1\rangle |+\rangle)$. Measurements in the $X$ basis 
\begin{equation} \label{mea1}
\{|+\rangle = \frac{1}{\sqrt{2}} (|0\rangle + |1\rangle), |-\rangle = \frac{1}{\sqrt{2}}(|0\rangle - |1\rangle)\}
\end{equation}
correspond to $\textbf{u}^i = 0$ and measurements in the $Z$ basis 
\begin{equation} \label{mea2}
\{|0\rangle, |1\rangle\} 
\end{equation}
correspond to $\textbf{u}^i = 1$ for each of the four parties $i \in \{1, \dots, 4\}$. These measurements on $|\Psi \rangle$ lead to the algebraic violation of the inequality, i.e., the sum of the probabilities appearing in the inequality is zero. 

The reason for the choice of this Bell inequality is twofold. Firstly, as we have seen, there exist quantum correlations achieving the maximal no-signaling violation of the inequality, which implies that free randomness amplification starting from any initial $\varepsilon$ of the SV source may be possible. Secondly, we will show (in Lemma \ref{lin-prog}) that, for any measurement setting $\textbf{u}$ out of the $2^4$ possible settings in the inequality, the probability of any of the $2^4$ output bit strings $\textbf{x}$ is bounded away from one (for any no-signaling box) by a linear function of the uniform value of the Bell expression. 

\section{General Set-up, Assumptions, and Randomness Criterion}
\label{sec:assumptions}
We have the following general setup. We consider the variable $\sv$ from the SV-source 
and the box held by the honest parties and Eve. The devices held by the honest parties have input and output denoted by $\abin$ and $\about$ and Eve's input and output are denoted $\ein$ and $\eout$, respectively. 

The devices held by the honest parties are separated into $m$ components with corresponding inputs and outputs $\abin^i$ and $\about^i$ for $i \in [m]$; furthermore, for each device component we have $l$ sequential runs with inputs and outputs \tgr{denoted by $\abin^i_{k}$ and $\about^i_k$ for $k\in[l]$, respectively.} The variable $\sv$ is divided into three parts $(\sve,\svin,\svhash)$,
where $\sve$ is generated by Eve, while $\svin$ and $\svhash$ are generated by the honest parties. 
The variables $\svin$  will be used as inputs to devices. 

The honest parties will check whether $\about$ and $\abin$ satisfy some 
joint acceptance condition (based on the violation of a particular Bell inequality); this acceptance event is denoted by $\ACC$. Conditioned on acceptance, the variables $\svhash$ together with the outputs $\about$ of devices will be hashed (using a multi-source extractor) to produce the output $s$ - the output random bits. The variables $\svin$ and $\svhash$ as well as the outputs $\about$ of the devices will be destroyed in the protocol (and in particular will remain hidden from Eve). 

The initial state of the system will be defined as the following correlation box: 
\be
p(\about,\eout,\svin,\svhash,\sve|\abin,\ein)
\ee

We assume the following:
\begin{itemize}
\item {\bf No-signaling assumptions:} 
The box satisfies the constraint of no-signaling between the honest parties and Eve 
\ben 
p(\about|\abin, \ein) &=& p(\about|\abin), \nonumber \\
p(\eout|\abin, \ein) &=& p(\eout|\ein),
\label{eq:as-nosig}
\een
as well as a no-signaling condition between device components
\ben
p(\about^I|\abin) = p(\about^I|\abin^I) \; \; \; \forall I \subseteq [m].
\een
Each device component also obeys a time-ordered no-signaling (\texttt{tons}) condition for the $k \in [l]$ runs performed on it:
\ben
p(\about^i_{k}|\eout,\abin^i,\ein,\sv) = p(\about^i_{k} | \eout,\abin^i_{\leq k},\ein,\sv) \; \; \; \forall k \in [l]
\label{eq:tons}
\een 
where $\abin^i_{\leq k} := \abin^i_{1}, \ldots, \abin^i_{k}$.

\item {\bf SV conditions:} 
The variables $(\svin,\svhash,\sve)$ form an SV source, that is satisfy Eq. (\ref{SVdef}). In particular,
$(\svhash|\svin,\sve)$ is also an SV source.

\end{itemize}

\begin{itemize}
\item {\bf Assumption A1:} The devices do not signal to the SV source, i.e. the distribution of $\sv$ is 
independent of the inputs $(\abin,\ein)$:
\be
\sum_{\about,\eout} p(\about,\eout,\sv|\abin,\ein) = p(\sv) \; \; \; \forall{(\about, \eout, \sv, \abin,\ein)}.
\label{eq:assumption1}
\ee

\item {\bf Assumption A2:} The box is fixed independently of the SV source:
\be
p(\about,\eout|\abin,\ein,\sv) = p(\about,\eout|\abin,\ein) \; \; \; \forall{(\about, \eout, \sv, \abin,\ein)}.
\label{eq:assumption2}
\ee
\end{itemize}

Assumption A1 is not v ery restrictive and is required to meaningfully describe the process of inputting a variable from the  SV source into the box. Assumption (A2), in turn, is more restrictive, although in our view is still a natural one. It is the quantum analogue of the problem of extracting randomness from one source and an independent (but unknown and arbitrary) channel. While it is easy to see that no randomness can be extracted in this classical setting, the situation is different considering non-local correlations. This assumption was also employed in the pioneering results \cite{Renner,Acin,Pawlowski}, as well as in all other results on the topic apart from \cite{CSW14}. 
%

After $\svin$ is input as $\abin$, we obtain the following box, 	
\be
p(\about,\eout,\svin,\svhash,\sve| \ein) :=  p(\about,\eout,\svin,\svhash,\sve|\svin, \ein)
\ee
Due to assumption A1, which assures in particular no-signaling from $\abin$ to $\svin$, this is a normalized probability distribution.

Conditioning on acceptance $\ACC$ and applying a hash function $s(\about,\svhash)$, one gets the following box
\begin{eqnarray}
\begin{aligned}
r(s,\eout,\sve|\ein) :=& \hspace{0.1 cm} p(s,\eout,\sve|\ein,\ACC )\\
\equiv &\sum_{\svin}\sum_{s(\about,\svhash)=s}   p(\about,\eout,\svin,\svhash,\sve,|\ein,\ACC )
\end{aligned}
\end{eqnarray}
where 
\be
p(\about,\eout,\svin,\svhash,\sve|\ein, \ACC)= 
\frac{p(\about,\eout,\svin,\svhash,\sve|\ein )}{\sum_{(\about,\svin)\in\ACC}p(\about,\svin|\ein )}
\ee

The composable security criterion \cite{RennerKoenigcomposability,Mayers} is defined in terms of the distance of $r(s,\eout,\sve|\ein )$ to an ideal box
$r^{id}=\frac{1}{|S|} r(\eout,\sve|\ein )$, with $r(\eout,\sve|\ein )=\sum_s r(s,\eout,\sve|\ein )$.
The distance is the standard variational distance between probability distributions
maximized over all possible measurements applied to the box. Since the most general measurement 
that Eve can apply is to look at the register $\sve$, as well as the register $s$, and input $\ein$ that may depend on both of them, we have:
\be
d_c=\sum_{s,\sve} \max_{\ein } \sum_{\eout} \left |r(s,\eout,\sve|\ein ) - \frac{1}{|S|}r(\eout,\sve|\ein )\right|.
\ee
One can rewrite it as follows
\be
d_c=\sum_{s,\sve}\max_{\ein }\sum_{\eout}  r(\eout,\sve|\ein )\left|r(s|\eout,\sve, \ein ) - \frac{1}{|S|}\right|
\ee
Rewriting it in terms of the box $p$ we obtain
\be
\label{eq:true_dc}
d_c= \sum_{s,\sve} \max_{\ein } \sum_{\eout } p(\eout, \sve |\ein ,\ACC) \left|p(s|\eout, \sve, \ein,  \ACC ) - \frac{1}{|S|}\right|.
\ee
The composable secure definition says that the protocol is $\varepsilon$-secure if $d_c \leq \varepsilon$, for a chosen error $\varepsilon$. It guarantees that even if part of $s$ is given to Eve, the rest is still secure.

We now note that 
\be
\label{eq:d_c_intro}
d_c\leq |S| \sum_{\sve} \max_{\ein } \sum_{\eout } p(\eout, \sve|\ein,\ACC) \sum_s \left|p(s|\eout,\sve, \ein, \ACC ) - \frac{1}{|S|}\right|.
\ee
In our proofs we will use this estimate and therefore will have to handle the extra factor of the size of the output $|S|$ in the error. The benefit will be that one can fix the measurement of Eve beforehand. In particular this allows us to use an extractor sound only with respect to a classical adversary and nevertheless obtain security under a no-signaling adversary. 

We note that assuming Eve is quantum, in other problems such as quantum key distribution and randomness expansion, it is possible to avoid the increase of the output error by $|S|$ using extractors that are sound against quantum side-information.  The fact that an extractor  that outputs $n$ bits with error $\varepsilon$ sound against classical side information is also sound against quantum (or even no-signaling) side information with error $2^n \varepsilon$ is well known (see e.g. \cite{Berta14}). However in many applications this error blow up is prohibitive. It turns out that in our approach we can afford it by tracing out a large fraction of the number of output bits.



\section{Tools for randomness amplification}
\label{Toolbox}

In this section we give the tools we will employ proving the correctness of Protocols I and II. 

\subsection{Estimation of the Bell value} 
\label{subsec:Azuma}
In this section we first show, that 
with high probability, the arithmetic average of mean values for conditional boxes is close to the observed value.  
As a corollary, we obtain that if the average is small, then with high probability for a linear fraction of 
all boxes the mean will be small too.  It has a similar flavor to previous results \cite{Pironio, Pironio13, Fehr11} 
obtained in the context of the related problem  of randomness expansion. 
\begin{remark}
\tgr{In the proof of Protocol I the variables $W_i$ will be interpreted as $W_i = (x_i,u_i)$,
for $i\geq 1$, where any $x_i$ and $u_i$ are of the form of $\bx=(\bx^1,\ldots,\bx^4)$ and $\bu=(\bu^1,\ldots,\bu^4)$, respectively, both introduced in Subsection II B, and  $W_0 = (z,e)$.  As far as we consider the general set-up described in Sec. III, we just put $\about:=x_i$, $\abin:=u_i$, $\eout:=z$ and $\ein:=e$ (see Sec. V for details).
}

\tgr{As for Protocol II, the lemma will be applied twice. For the first time, 
we will take  $W_i=(x_i,u_i)$ for $i\geq1$ and 
 $W_0= (z,v^j, M_j, e)$. For the second time, we will take  $W_i=(y^j_i,v^j_i)$ for $i\geq1$ and 
 $W_0= (z,u, M_j, e)$ (see Sec. VI for the definitions of all variables used above).
 }
 
\tgr{
The function $B_i$ will be the same for both those cases, and  for all $i\geq 1$, given by \eqref{eq:Bell_indicator}.
 }
 \end{remark}

\begin{lemma} \label{lemmaazuma} 
Consider  arbitrary random variables $W_i$, for $i=0,1,\ldots,n$, and binary random variables $B_i$, for $i=1,\ldots n$, that are functions of $W_i$, i.e. $B_i=f_i(W_i)$ 
for some functions $f_i$. Let us denote $\overline{B}_i=\mathbb{E}(B_i|W_{i-1},\ldots,W_1,W_0)$, for $i=1,\ldots,n$ (i.e. $\overline{B}_i$ are conditional expectation values).
Define, for $k = 1, \ldots, n$, the empirical average
\be
\Lazuma_k=\frac1k\sum_{i=1}^k B_i
\ee
and the arithmetic average of conditional expectation values 
\be
\label{aver-cond-mean}
\overline{\Lazuma}_k=\frac1k\sum_{i=1}^k \overline{B}_i.
\ee
Then we have 
\be
P(|\Lazuma_n-\overline{\Lazuma}_n|\geq s)\leq 2 e^{-n\frac{s^2}{2}}.
\label{eq:poor_man}
\ee
\label{lem:poor_man}
\end{lemma}


To prove lemma \ref{lem:poor_man}, we need to state the Azuma-Hoeffding inequality.
Let $X_0, \ldots, X_k$ and $W_0, \ldots, W_k$ be two sequences of random variables. Then $X_0, \ldots, X_k$ is said to be a
martingale with respect to $W_0, \ldots, W_k$ if for all $0 \leq i \leq k$, $\mathbb{E}|X_{i}| < \infty$ and $\mathbb{E}(X_{i} | W_0, \ldots, W_{i-1}) = X_{i-1}$.

\begin{lemma} 
(Azuma-Hoeffding) Suppose $X_0, \ldots, X_k$ is a martingale with respect to $W_0, \ldots, W_k$, and that $|X_{l+1} - X_l| \leq c_l$ for all $0 \leq l \leq k-1$. Then, for all positive reals $t$, 
\begin{equation}
P \left( |X_k - X_0| \geq t  \right) \leq 2\exp \left( - \frac{t^2}{2 \sum_{l=1}^{k} c_l^2}  \right).\\
\end{equation}
\label{lem:azuma}
\end{lemma}

Now we can prove Lemma \ref{lem:poor_man}.
\begin{proof}
Define $X_0=0$, $X_l=l(\Lazuma_l-\overline{\Lazuma}_l)$. Let us show that $\{X_i\}_{i=0}^n$ and $\{W_i\}_{i=0}^n$
satisfy the assumptions of Lemma \ref{lem:azuma}. First, 
\be
|X_l-X_{l-1}|=|\overline{B}_l-B_l|\leq 1,
\ee
since $B_l$ is binary and $0\leq\overline{B}_l\leq 1$. Let us now check that 
$\{X_l\}_{l=0}^n$ is a martingale with respect to $\{W_l\}_{i=0}^{n-1}$. We have $\mathbb{E}|X_{l}|\leq l < \infty$.
Moreover, for $l\geq 2$ we have
\ben
\EE(X_l|W_{l-1},\ldots,W_0)&=& \sum_{i=1}^l\EE(B_i|W_{l-1},\ldots,W_0)-
\sum_{i=1}^l \EE(\overline{B}_i |W_{l-1},\ldots,W_0) \\ 
&=& \sum_{i=1}^l \EE(B_i |W_{l-1},\ldots,W_0) -\sum_{i=1}^l \overline{B}_i   \\
&=& \sum_{i=1}^{l-1} \EE(B_i |W_{l-1},\ldots,W_0) - \sum_{i=1}^{l-1} \overline{B}_i 
\label{eq:mart}
\een
where we used the property 
\be
\label{eq:EABC}
\EE(\EE(A|B)|BC)=\EE(A|B)
\ee
for random variables $A,B,C$. 
Now let us note that since $B_i$ is a function of $W_i$, we get  for $i\leq l-1$
\be
\EE(B_i |W_{l-1},\ldots,W_1,W_0)=B_i
\ee
Thus the last line of Eq. \eqref{eq:mart} is  equal to $X_{l-1}$, so that we have 
\be
\EE(X_l|W_{l-1},\ldots,W_1,W_0)=X_{l-1}
\label{eq:mart2}
\ee
for $l=2,\ldots,n$. For $l=1$ one may verify Eq.\eqref{eq:mart2} by checking directly that $\EE(X_1|W_0)=0$,
being hence equal to $X_0$ defined to be zero. \tgr{Indeed, in the latter case $X_1 = B_1 - \overline{B}_1$ where $\overline{B}_1 = \EE(B_1|W_0)$, hence using \eqref{eq:EABC} we get 
$\EE(B_1-\overline{B}_1|W_0)= \EE(B_1|W_0) - \EE(\EE(B_1|W_0)|W_0) = 0$}. Now, we apply Lemma \ref{lem:azuma} with $c_l=1$, and obtain the inequality \eqref{eq:poor_man}.
\end{proof}

We also note the following useful fact. 
\begin{lemma}
\label{lem:linear_fraction}
If the arithmetic average $\overline{\Lazuma}_n$ of $n$ conditional means in Eq.(\ref{aver-cond-mean}) satisfies  $\overline{\Lazuma}_n\leq \delta$ for some parameter $\delta > 0$, then 
in at least $(1-\sqrt{\delta}) n$ of positions $i$ we have $\overline{B}_i\leq \sqrt{\delta}$
\end{lemma}

\begin{proof}
Assume that $\overline{\Lazuma}_n = \frac{1}{n} \sum_{i=1}^n \overline{B}_i \leq \delta$ with $\overline{B}_i \geq 0 \; \; \forall i$. Consider the set $I :=\{ i | \overline{B}_i \geq \frac{\delta}{\gamma} \}$. Then, $\frac{|I|}{n} \frac{\delta}{\gamma} \leq \delta$ so that $|I| \leq n \gamma$. Choosing $\gamma = \sqrt{\delta}$ we obtain that the fraction $\mu$ of positions $i$  with value $\overline{B}_i < \sqrt{\delta}$ is given by $\mu = (1 - \frac{|I|}{n}) \geq 1 - \sqrt{\delta}$.  
\end{proof}

\subsection{Randomness of individual box from good Bell value} 
\label{subsec:SV-Bell}

In Sec. \ref{subsec:Azuma} we have shown that if the observed Bell value is small, then there is linear 
number of conditional boxes, with small Bell value (with settings chosen from a SV source).
In Sec. \ref{subsec:min-entropy-source} we will show that in order to obtain a min-entropy source we need to ensure that a~constant fraction of the conditional boxes has randomness. 
In this section we will tie up the two observations by arguing that the randomness of a box is ensured if the value of the Bell expression with inputs taken from a SV source is small. 

Let ${U}$ denote all the settings appearing in the Bell expression. We consider first the uniform Bell value
\be
\label{uniform-Bell}
\overline{B}^U := \frac{1}{|\textsl{U}|} \textbf{B}. \{ P(\textbf{x} | \textbf{u}) \} = \frac{1}{|\textsl{U}|} \sum_{\bx,\bu}  B(\bx,\bu) P(\bx|\bu),
\ee
where $|\textsl{U}|$ denotes the cardinality of ${U}$, i.e. the total number of settings in the Bell expression. If the Bell function $B(\bx,\bu)$ is properly chosen, one can prove using linear programming that if $\overline{B}^U$ is small, the probabilities of any outputs are bounded away from 1. However, since our inputs to each device are chosen using a SV source, we will be only able to estimate the value of the following expression
\be
\label{SV-Bell}
\overline{B}^{SV}= \sum_{\bx,\bu} \nu(\bu)  B(\bx,\bu) P(\bx|\bu),
\ee
where $\nu(\bu)$ is the distribution from an (unknown) SV source. We will show that for a suitably chosen Bell function, when the latter expression is small, the former is also small which implies randomness. 

In the following lemma, we prove the relation between SV Bell value and randomness for a~particular Bell inequality given by Eq. \eqref{Bell-ineq}. It says that for SV source of arbitrary $\varepsilon \not= \frac12$,
if the SV Bell value is small enough, the probability of any outcome is bounded away from 1.
\begin{lemma}
\label{lem:B_SV_rand}
Consider a four-partite no-signaling box $P (\textbf{x}| \textbf{u})$ satisfying 
\be
\overline{B}^{SV}\leq \delta,
\ee
where $\overline{B}^{SV}$ is given by Eq. \eqref{SV-Bell} with $B(\bx,\bu)$ given by Eq. \eqref{eq:Bell_indicator}. 
Then, for any measurement setting $\textbf{u}^*$ and any output $\textbf{x}^*$, we have
\begin{equation} \label{SV-output-bound}
P \left( \textbf{x}^* | \textbf{u}^*\right) \leq   \frac13 \left(1+ \frac{2 \delta}{(\frac12 -\epsilon)^4 }  \right).          
\end{equation}\end{lemma}

\begin{proof}
From the definition of an $\varepsilon$-SV source we have 
\begin{equation}
(\frac{1}{2} - \varepsilon)^4 \leq \nu(\bu) \leq (\frac{1}{2} + \varepsilon)^4.
\end{equation}
so that 
\begin{equation}
\label{SV-uniform}
\frac{1}{(\frac{1}{2} + \varepsilon)^4 |\textsl{U}|} \overline{B}^{SV} \leq \overline{B}^{U} \leq \frac{1}{(\frac{1}{2} - \varepsilon)^4 |\textsl{U}|} \overline{B}^{SV}
\end{equation}
Then the claim follows from lemma \ref{lin-prog}, relating  $\overline{B}^{U}$ with 
$P \left( \textbf{x}^* | \textbf{u}^* \right)$ by use of linear programming. 
\end{proof}

\subsubsection{Bounding output probabilities by linear programming}

Let us show that for the specific Bell inequality we consider, when the value of the Bell expression is small there is weak randomness. Consider a four-partite no-signaling box $P(\textbf{x} | \textbf{u})$ that obtains a value $\delta$ for the Bell expression in Eq. \eqref{Bell-ineq}. The following lemma shows that for any measurement setting $\textbf{u}$, the probability of any outcome $\textbf{x}$ is bounded from above by a function of~$\delta$.

\begin{lemma} \label{lin-prog}
Consider a four-partite no-signaling box $P (\textbf{x}| \textbf{u})$ satisfying 
\begin{equation} 
\label{goodindividual}
\overline{B}^U := \frac{1}{|\textsl{U}|} \textbf{B}. \{ P(\textbf{x} | \textbf{u}) \} \leq \frac{\delta}{|\textsl{U}|},
\end{equation}
for some $\delta \geq 0$, with $\textbf{B}$ the indicator vector for the Bell expression in Eq. \eqref{Bell-ineq} and $|\textsl{U}| = 16$ the number of settings in the Bell expression. For any measurement setting $\textbf{u}^*$ and any output $\textbf{x}^*$, we have
\begin{equation} \label{prod-dist}
P \left( \textbf{x}^* | \textbf{u}^*\right) \leq  \frac{1+ 2 \delta}{3}.          
\end{equation}
\end{lemma}

\begin{proof}
Consider any measurement setting $\textbf{u}^*$ and any corresponding output $\textbf{x}^*$ for this setting. 
Then $P \left( \textbf{x}^* | \textbf{u}^* \right)$ can be computed by the following linear program
\begin{eqnarray}
\label{lin-prog1}
P \left( \textbf{x}^* | \textbf{u}^* \right) &=& \max_{ \{ P \}}: \textit{M}_{\textbf{x}^*, \textbf{u}^*}^T \{ P(\textbf{x} | \textbf{u}) \} \nonumber \\
&&s.t. \; \; \textit{A} \{ P( \textbf{x} | \textbf{u}) \} \leq \textit{c}.
\end{eqnarray}
Here, the indicator vector $\textit{M}_{\textbf{x}^*, \textbf{u}^*}$ is a $2^4 \times 2^4$ element vector with entries 
$M_{\textbf{x}^*, \textbf{u}^*}(\textbf{x}, \textbf{u}) = \texttt{I}_{\textbf{x} = \textbf{x}^*} \texttt{I}_{\textbf{u} = \textbf{u}^*}$. 
The constraint on the box $\{P(\textbf{x} | \textbf{u})\}$ written as a vector with $2^4 \times 2^4$ entries is given by the matrix $\textit{A}$ and the vector $\textit{c}$. These encode the no-signaling constraints between the four parties, the normalization and the positivity constraints on the probabilities $P(\textbf{x} | \textbf{u})$. In addition, $\textit{A}$ and $\textit{c}$ also encode the condition that $\textbf{B}.\{ P(\textbf{x}| \textbf{u}) \} \leq \delta$ with $\delta$ the bound on the Bell value for the box. Analogous programs can be formulated for each of the $2^4$ measurement settings appearing in the Bell inequality in Eq. (\ref{Bell-ineq}) and each of the $2^4$ corresponding outputs.

The solution to the primal linear program in Eq. (\ref{lin-prog1}) can be bounded by any feasible solution to the dual program which is written as
\begin{eqnarray}
\label{dual-lin-prog1}
&&\min_{ \lambda_{\textbf{x}^*, \textbf{u}^*}}: \textit{c}^T \lambda_{\textbf{x}^*, \textbf{u}^*} \nonumber \\
&& s.t. \; \; \; \textit{A}^T  \lambda_{\textbf{x}^*, \textbf{u}^*} = \textit{M}_{\textbf{x}^*, \textbf{u}^*}, \nonumber \\
&&\; \; \; \; \; \; \; \;  \lambda_{\textbf{x}^*, \textbf{u}^*} \geq 0.
\end{eqnarray}
For each $\{\textbf{u}^*, \textbf{x}^*\}$, we find a feasible $\lambda_{\textbf{x}^*, \textbf{u}^*}$ satisfying the constraints to the dual program above that gives $\textit{c}^T \lambda_{\textbf{x}^*, \textbf{u}^*} \leq \left( \frac{1 + 2 \delta}{3} \right)$. We therefore obtain by the duality theorem of linear programming that
\begin{equation}
P \left( \textbf{x}^* | \textbf{u}^* \right) \leq  \left( \frac{1+ 2 \delta}{3}  \right),
\end{equation}
which is the required bound.
\end{proof}

\subsection{A min-entropy source from randomness of conditional boxes}
\label{subsec:min-entropy-source}

In this section we show that if a device is such that a linear number of conditional boxes have randomness (in the weak sense that the probability of their outputs is bounded away from one), then the distribution on outputs constitutes a min-entropy source. The considerations in this section will be applicable to any of the devices $j \in [k]$ and any chosen block. Therefore we will skip the indices for simplicity. 

Let any sequence $(\xuseq)$ be such that $x_i$ and $u_i$, $i\in \{1,\ldots,n\}$, are of the form of $\bx=(\bx^1,\ldots,\bx^4)$ and $\bu=(\bu^1,\ldots,\bu^4)$, respectively, both introduced in Subsection II B. We will show that if, with large probability over sequences $(\xuseq)$, a constant fraction of those boxes has, for any setting, probability of every output bounded away from $1$, then the total probability distribution is close in variational distance to a min-entropy source (see \cite{Pironio, Pironio13, Fehr11} for a similar result in the context of randomness expansion).

We first prove that, if this happens for {\it all} sequences (i.e. with probability $1$), 
then the total box is a min-entropy source itself and subsequently consider the case when the probability is close to~$1$. 

\begin{lemma}
\label{lem:min-entropy-0}
Fix any measure $P$ on the space of sequences $(\xuseq)$. 
Suppose that for a sequence $(\xuseq)$, there exists $\texttt{K} \subseteq [n]$ of size larger than $\mu n$, such that for all $l \in \texttt{K}$ the conditional boxes $\pxu$ satisfy 
\be
\pxu\leq \gamma,
\label{eq:random-box2}.
\ee
Then, 
$\pxutot$ satisfies 
\begin{equation}
\pxutot\leq \gamma^{\mu n}
\label{eq:hmin_cond}
\end{equation}
\end{lemma}

\begin{proof}
The proof proceeds by successive application of the Bayes rule and the time-ordered no-signaling structure, i.e. 
\ben
\pxutot&=&P(x_1|u_1)P(x_2|u_2, x_1,u_1)\nonumber \\
&&\ldots P(x_n|u_n, x_{n-1},u_{n-1},\ldots,x_1,u_1),
\een
where we have used the fact that the outputs of the $l$-th box can depend only upon the inputs and outputs of the previous boxes due to the time-ordered structure of the boxes (see Eq. (\ref{eq:tons})). Now, due to the assumption that at least 
$\mu n$ of the conditional boxes $P_{{x}_{< l}, {u}_{< l}}({x}_l |{u}_l)$ 
satisfy Eq. \eqref{eq:random-box2}, we have that 
\begin{equation}
\pxutot \leq \gamma^{\mu n}.
\end{equation}
\end{proof}

\subsection{Imposing independence between devices by a de Finetti bound with limited randomness}
\label{deFinetti-section}
Consider two devices, the first consisting of $n$ boxes and the second consisting of $N_2$ blocks of $n$ boxes each. In this section, we show that, for suitable choice of $N_2$, the boxes from the first device are close to being uncorrelated with the boxes in a block chosen from the second device using an $\varepsilon$-SV source. The lemma is based on the information-theoretic approach of \cite{Brandao, Brandao2} for proving de Finetti theorems for quantum states and no-signaling distributions. 

We denote the box by $P({X}^1, {X}_{\leq N_2}^2 | {U}^1, {U}_{\leq N_2}^2)$, where the superscript denotes the device and the subscript denotes the block of uses of the device. Capital letters denote the inputs and outputs for a set of $n$ boxes so that ${X}^1 = ({x}^1_1, \ldots,{x}^1_{n})$ and ${X}^2_{\leq N_2} = ({x}^2_{1,1}, \ldots,{x}^2_{n,1}, \ldots,{x}^2_{1, N_2}, \ldots, {x}^2_{n, N_2})$ with the second subscript denoting the block. Note that any $x^1_k$ or $x^2_{k,l}$ for $k\in\{1,\ldots,n\}$, $l\in\{1,\ldots,N_2\}$ are of the form $\bx=(\bx^1,\ldots,\bx^4)$ introduced in Sec. II B. Similarly $u^1_k$ and $u^2_{k,l}$, $k\in\{1,\ldots,n\}$ for any $l\in\{1,\ldots,N_2\}$
 are of the form  $\bu=(\bu^1,\ldots,\bu^4)$.

\begin{lemma}
\label{lem:deFinetti-bound1}
Let $P({X}^1, {X}_{\leq N_2}^2 | {U}^1, {U}_{\leq N_2}^2)$ 
satisfy the no-signaling conditions, i.e. Eqs. (\ref{eq:as-nosig})-(\ref{eq:tons}),
with output and input alphabets $\Sigma$ and $\Lambda$, respectively (i.e. $P : \Sigma^{\times (N_2 + 1)n} \times \Lambda^{\times (N_2 + 1)n} \rightarrow \mathbb{R}^+$). The distribution $P$ represents two devices with the first containing $n$ boxes and the second $N_2$ blocks of $n$ boxes each. Let $A_2 \in [N_2]$ and $({U}^1, {U}_{\leq N_2}^2)$ be chosen from an $\varepsilon$-SV source\tgr{; we write $\nu(j, {U}^1,{U}^2_{\leq N_2})$ and the distribution $\nu$ satisfies condition (\ref{SVdef})}. Then, we have
\begin{eqnarray}
\label{deFinetti}
&&\mathop{\mathbb{E}}_{ (j,{U}^1, {U}^2_{\leq N_2}) \sim \nu} \hspace{0.1 cm}  \mathop{\mathbb{E}}_{{X}^{2}_{< j} \sim P(.|{U}^1,{U}^2_{\leq N_2} )} \left \Vert \tilde{P}({X}^1,{X}^{2}_{j} |{U}^1, {U}^{2}_{j}) - \tilde{P}({X}^1 | {U}^1) \otimes \tilde{P}({X}^{2}_{j} | {U}^2_{j}) \right \Vert \nonumber \\
&& \hspace{6cm} \leq \sqrt{2 \ln{(2)} N_2^{\log{(1 + 2 \varepsilon})} \frac{n \log{|\Sigma|}}{N_2}},
\end{eqnarray}
where $\tilde{P}$ is the conditional box given the inputs ${U}_{< j}^2$ and outputs $\textbf{X}_{< j}^2$ of all prior boxes to the ones in the $j$-th block, i.e.
\begin{equation}
\label{cond-box-a}
\tilde{P}({X}^1, {X}^2_{j} | {U}^1, {U}^2_{j}) := P_
{{X}_{< j}^2, {U}_{< j}^2} ({X}^1, {X}^2_{j} |{U}^1, {U}^2_{j}).
\end{equation}
\end{lemma}

\begin{proof}
Using the upper bound on mutual information $I(A:B) \leq \min (\log{|A|}, \log{|B|})$ and the chain rule $I(A: B C) = I(A:B) + I(A:C | B)$, we have that for every distribution \tgr{$\nu$}
\begin{eqnarray}
n \log |\Sigma| &\geq &  \mathop{\mathbb{E}}_{{U}^1, {U}^2_{\leq N_2} \sim \tgr{\nu}} I({X}^1 : {X}^2_{\leq N_2})_{\tgr{P(\cdot|{U}^1, {U}^2_{\leq N_2})}} \nonumber \\
&=& \mathop{\mathbb{E}}_{{U}^1, {U}^2_{\leq N_2} \sim \tgr{\nu}} 
\left( I({X}^1 : {X}^2_1)_{\tgr{P(\cdot|{U}^1, {U}^2_{\leq N_2})}} 
+ \ldots + I({X}^1 : {X}^2_{N_2} | {X}^2_{< N_2})_{\tgr{P(\cdot|{U}^1, {U}^2_{\leq N_2})}} \right) \nonumber \\
&=&  \mathop{\mathbb{E}}_{{U}^1,{U}^2_{\leq N_2} \sim \tgr{\nu}} \mathbb{E}_{j \sim U(N_2)} I({X}^1 : {X}^2_{j} | {X}^2_{1}, \ldots, {X}^2_{j-1})_{\tgr{P(\cdot|{U}^1, {U}^2_{\leq N_2})}}\nonumber \\
&=&  \mathop{\mathbb{E}}_{{U}^1,{U}^2_{\leq N_2} \sim \tgr{\nu}} \mathbb{E}_{j \sim U(N_2)} I({X}^1 : {X}^2_{j} )_{\tgr{\tilde{P}(\cdot|{U}^1, {U}^2_j)}},
\end{eqnarray}
where $U(N_2)$ is the uniform distribution over the set $[N_2]$.

Therefore, if $j, {U}^1, {U}^2_{\leq N_2}$ are chosen from an $\varepsilon$-SV source $\nu$, we find 
\begin{equation} \label{mutual-k}
\begin{aligned}
\mathop{\mathbb{E}}_{(j, {U}^1, {U}^2_{\leq N_2}) \sim \nu}  
\mathop{\mathbb{E}}_{{X}^1, {X}^{2}_{< j} \sim P} &I({X}^1 : {X}^{2}_{A_2})_{\tgr{\tilde{P}(\cdot|U^1,U^2_j)}}\\
&=  \mathop{\mathbb{E}}_{({U}^1, {U}^2_{\leq N_2}) \sim \nu}  
\mathop{\mathbb{E}}_{j \sim \nu(.|{U}^1,{U}^2_{\leq N_2})}  
\mathbb{E}_{{X}^1, {X}^{2}_{< j} \sim P} I({X}^1 : {X}^{2}_{A_2})_{\tgr{\tilde{P}(\cdot|U^1,U^2_j)}}   \\  
&\leq  \left( \frac{1}{2} + \varepsilon \right)^{\log{N_2}}  
\mathop{\mathbb{E}}_{({U}^1, {U}^2_{\leq N_2}) \sim \nu} 
\mathop{\mathbb{E}}_{j \sim U(N_2)}  
\mathop{\mathbb{E}}_{{X}^1, {X}^{2}_{< j} \sim P} I({X}^1 : {X}^{2}_{j})_{\tgr{\tilde{P}(\cdot|U^1,U^2_j)}}  \\  
&\leq \left( \frac{1}{2} + \varepsilon \right)^{\log{N_2}} n \log |\Sigma|. 
\end{aligned}
\end{equation}
We now use Pinsker's inequality  relating the mutual information and trace distance 
for any measure $Q$ as 
\begin{equation}
\label{Pinsker}
I(A : B)_Q \geq \frac{1}{2 \ln{(2)}} \left \Vert Q(A,B) - Q(A) \otimes Q(B) \right \Vert^2
\end{equation}
with $Q:=\tilde{P}(\cdot|U^1,U^2_j)$, 
and the convexity of $x \mapsto x^2$ to obtain Eq. (\ref{deFinetti}).
\end{proof}



\section{Proof of correctness of Protocol I}
\label{sec:protocol1}

In Protocol I, the honest parties and Eve share a no-signaling box $P(x,z|u',w)$, where $(x, u')$ denotes the outputs and inputs of all the honest parties for the $n$ runs of the protocol and $(z, w)$ denotes the output and input of the adversary Eve. The honest parties obtain bits $u$ from the SV source that will serve as inputs to their box, so $u'$ will be set to be equal to $u$. They also draw further bits $t$ from the SV source to feed together with $x$ into the randomness extractor, obtaining the final 
output of the protocol 
\be
\label{eq:s_protocolI}
s=s(x,t).
\ee
Eve has classical information $e$ which are bits correlated to $u$, $t$. The initial box, describing all initial variables and inputs is given by
\be
p(x,z,u,t,e| u', w)
\ee
which is a family of probability distributions labelled by $u'$ and $w$
(we also denote them by $p_{u', w}(x,z,u,t,e)$). 
The described setup is illustrated in Fig. \ref{fig:protocol1}. 
\begin{figure}
\scalebox{0.40}{\includegraphics{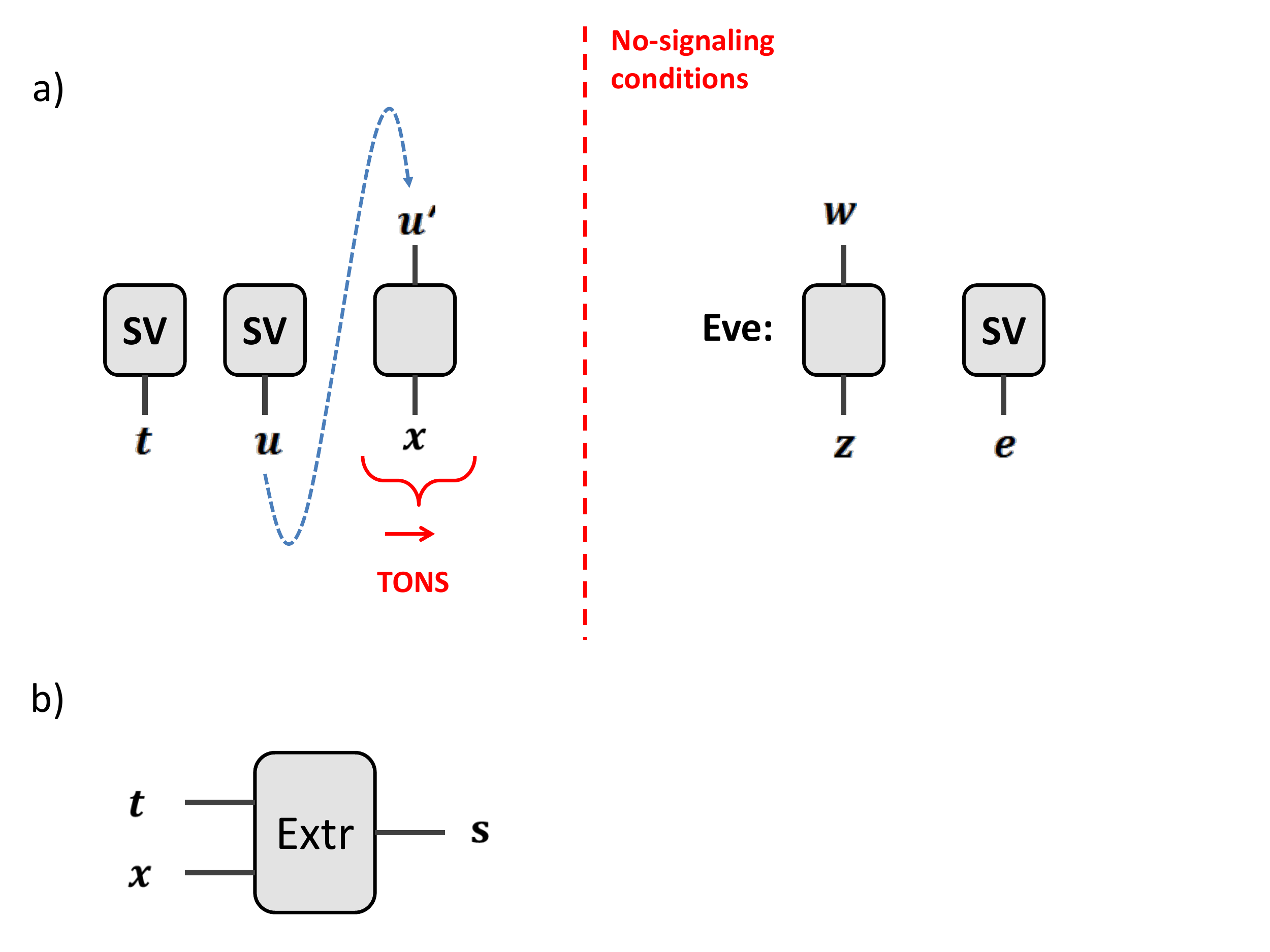}}
\caption{The setup for protocol I.} 
\label{fig:protocol1}
\end{figure}
In notation from Sec. \ref{subsec:min-entropy-source}, we have $x=(x_1,\ldots,x_n)$, $u=(u_1,\ldots, u_n)$ where every $x_k$ and $u_k$ for $k\in\{1,\ldots,n\}$ is of the form $\bx=(\bx^1,\bx^2,\bx^3,\bx^4)$ and $\bu=(\bu^1,\bu^2,\bu^3,\bu^4)$, respectively.

Now, referring to notation in Sec. \ref{sec:assumptions},  for the  present protocol, 
we have $\svin=u$, $\svhash=t$, $\sve=e$, $\about=x$, $\abin=u'$, $\ein=w$, $\eout=z$. 
The assumptions of Sec. \ref{sec:assumptions} \tgr{thus read} as follows.

\vspace{0.2 cm}

\begin{itemize}

\item {\bf No-signaling assumptions:}
\tgr{\ben 
\label{eq:NS_I1}
p(x|u',w) &=& p(x|u'),  \\
\label{eq:NS_I2}
p(z|u',w) &=& p(z|w),  \\
\label{eq:NS_Itons}
p(x_{j}|u',w,z,u,t,e) &=& p(x_{j} | u'_{\leq j},w,z,u,t,e) \; \; \; \forall j \in [n].
\een }

\item {\bf Assumption A1-(I):} The devices do not signal to the SV source, 
i.e. the distribution of $(u, t, e)$ is independent of the inputs $(u', w)$:
\be
\sum_{x,z} p(x,z,u,t,e|u',w) = p(u,t,e) \; \; \; \forall{(u,t,e,u',w)}.
\label{eq:assumption1}
\ee

\item {\bf Assumption A2-(I):} The box is fixed independently of the SV source:
\be
p(x,z|u',w,u,t,e) = p(x,z|u',w) \; \; \; \forall{(x,z,u,t,e,u',w)}.
\label{eq:assumption2}
\ee

\item {\bf SV conditions:} 
\tgr{The distribution $p(u,t,e)$ satisfies an SV condition (\ref{SVdef}); in particular,
$p(t|u,e)$ satisfies Eq. (\ref{SVdef}) too.}
\end{itemize}
After inputting $u$ as $u'$  (as is done in the protocol) we obtain 
\be
\label{eq:p_w}
p_w(x,z, u, t, e) := p(x,z,u, t, e|u, w).
\ee
We note that, \tgr{due to Assmuption A1-(I)} (i.e. no-signaling from input $u'$ to the variable $u$), it is a~normalized probability distribution for every $w$. 

Now, for
\be
L(x,u)= \frac1n \sum_{i=1}^n B_i(x_i,u_i),
\ee
we define the sets $\ACCd$ and $\ACCd_u$ for acceptance of the protocol as follows
\be
\ACCd = \{(x,u): L(x,u) \leq \delta \},
\ee
and 
\be
\ACCd_u = \{ x : (x,u) \in \ACCd\}.
\ee
Upon acceptance of the protocol, the family of probability distributions \eqref{eq:p_w} is modified to $p_w(x,z,u,t,e|\ACCd)$.

\tgr{To quantify the quality of the output $s$}, we will use the universally composable distance defined in  \eqref{eq:true_dc}
which in this case reads as
\begin{equation}
\label{eq:dcintro_I}
\dcintro =  \sum_{s,e} \max_w \sum_{z} \left| p_w(s,z,e|\ACCd) - \frac{1}{|S|} p_w(z,e|\ACCd) \right|. 
\end{equation}
with $s$ given by \eqref{eq:s_protocolI}.
Here the probability distributions  $p_w(s,z,e|\ACCd)=\sum_u p_w(s,z,u,e|\ACCd)$ are computed from probability distributions \eqref{eq:p_w}.
Actually, in the proofs we will deal with slightly modified distance 
\begin{equation}
\label{eq:dist-w}
\dcours :=  \sum_{s,e} \max_w \sum_{z,u} \left| p_w(s,z,u,e|\ACCd) - \frac{1}{|S|} p_w(z,u,e|\ACCd) \right|. 
\end{equation}
By triangle inequality, we have 
\be
\label{eq:dintro_dours}
\dcintro\leq \dcours,
\ee
hence it is enough to bound $\dcours$.

We now define an auxiliary quantity
\be
\label{eq:dist-w}
\dwithw := \sum_{e} p(e|\ACCd) \max_w \sum_{z,u} p_w(z,u|e,\ACCd)  \sum_s \left| p_w(s|z,u,e,\ACCd) - \frac{1}{|S|}\right|
\ee
for any family of probability distributions $\{p_w(x,z,u,t,e)\}$.

\tgr{\begin{remark}\label{rem:implication}
In the last two sections we many times use the following, easy to prove, implication
\be\label{implication}
P(a|b,c,d)=P(a|b)\quad\Rightarrow\quad P(a|b,c,d)=P(a|b,c)=P(a|b,d).
\ee
where $P$ is an arbitrary probability measure. 
\end{remark}}
\tgr{From Assumption A1-(I) and A2-(I), as well as no-signaling assumptions (Eqs. (\ref{eq:NS_I1})-(\ref{eq:NS_Itons})),
we find that the distributions $\{p_w(x,z,u,t,e)\}$ satisfy:
\ben
\label{eq:p-cond1}
&&p_w(x,u)=p(x,u) \quad (\text{implied by Eq. (\ref{eq:NS_I1})} \; \text{and A1-(I)}),  \\
\label{eq:p-cond2}
&& p_w (u,t,e)=p(u,t,e)  \quad (\text{follows from A1-(I)}), \\
\label{eq:p-cond3}
&& \forall_w \hspace{0.1 cm} p_w(x,z|u,t,e)= p_w(x,z|u) \quad (\text{follows directly from A2-(I)}), \\
\label{eq:p-cond4}
&& \forall_w \hspace{0.1 cm} p_w(x,z|u,t,e)= p_w(x,z|u,e) \quad (\text{follows from A2-(I) and 
Rem. \ref{rem:implication}}), \\
\label{eq:p-cond5}
&& p_w(x|z,u,t,e) = p_{z,t,e,w}(x|u) \; \; \; \text{is time-ordered no-signaling box (by (\ref{eq:NS_Itons}))}, \\
\label{eq:p-cond6}
&& p_w(u|z,e) \; \text{and} \; p_w(t|z,u,e) \; \; \; \text{are SV sources} \; (\text{from A2-(I) and SV conditions above}).
\een}

\tgr{\begin{remark}\label{SV_proof}
To be precise, the proof of property (\ref{eq:p-cond6}) goes as follows. We know that $p(u,t,e)$ is an SV source. Note that
\be 
p_w(u|z)=\frac{p_w(u,z)}{p_w(z)}=\frac{\sum_{t,e}p_w(z,u,t,e)}{p_w(z)}
\stackrel{\text{A2-(II) \& Eq. (\ref{eq:NS_I2})}}{=}\frac{p_w(z)\sum_{t,e}p_w(u,t,e)}{p_w(z)}
\stackrel{\text{A1-(I)}}{=}\sum_{t,e}p(u,t,e)
\ee
and hence $p_w(u|z)$ is an SV source. Further, we also obtain
\be 
p_w(t|z,u,e)=\frac{p_w(z,u,t,e)}{\sum_t p_w(z,u,t,e)}
\stackrel{\text{A2-(I) \& Eq. (\ref{eq:NS_I2})}}{=}
\frac{p_w(z)p_w(u,t,e)}{p_w(z)\sum_t p_w(u,t,e)}=p_w(t|u,e)
\ee
and therefore, by the fact that $p(t|u,e)$ satisfies an SV condition (\ref{SVdef}), the distribution $p_w(t|z,u,e)$ also is an SV source.
\end{remark}}
For each $e$, let $w_e$ denote Eve's input $w$ and let $p_{w_{e}}(x,z,u,t|e)$ denote the corresponding distribution that achieves the maximum in Eq.(\ref{eq:dist-w}). Using the fact that $p_{w}(e) = p(e)$ and that $\ACCd$ is a~set of $(x,u)$ which obey $p_{w}(x,u) = p(x,u)$ 
(from Eq.(\ref{eq:p-cond1})), we see that the distribution that achieves the maximum takes the form $p(e) p_{w_{e}}(x,z,u,t|e)$. 
We now set 
\be
\label{eq:qvspw}
q(x,z,u,t,e) := p(e) p_{w_{e}}(x,z,u,t|e). 
\ee
It can be readily seen that this $q(x,z,u,t,e)$ obeys the restrictions: 
\ben
\label{eq:q-cond1}
&&q(x,z|u,t,e)= q(x,z|u), \\
\label{eq:q-cond2}
&&q(x,z|u,t,e)= q(x,z|u,e), \\
\label{eq:q-cond3}
&&q(x|z,u,t,e)=q_{t,e,z}(x|u) \; \; \; \text{is time ordered no-signaling box}, \\
\label{eq:q-cond4}
&&q(u|z,e) \; \text{and} \; q(t|z,u,e) \; \; \; \text{are SV sources}.
\een

We can therefore define 
\be
\label{eq:dist-no-w}
\dnow :=  \sum_{z,u,e} q(z,u,e|\ACCd) \sum_s \left| q(s|z,u,e,\ACCd) - \frac{1}{|S|}\right| 
\ee
and observe that 
\be
\label{eq:dnow_dwithw}
\dnow = \dwithw.
\ee 
We now have
\begin{prop}
\label{lem:dc-dmax}
For any distribution $p_w(x,z,u,t,e)$ given by (\ref{eq:p_w}) we have
\be
\dcintro \leq |S| \dnow.
\ee
\end{prop}
\begin{proof}
This is a consequence of (\ref{eq:dintro_dours}) and the following inequalities
\begin{eqnarray}
\label{eq:dc-dmax}
\dcintro\leq\dcours &=& \sum_{s, e} p(e|\ACCd) \max_w \sum_{z,u} p_w(z,u|e,\ACCd) \left| p_w(s|u,e,z, \ACCd) - \frac{1}{|S|} \right|
 \nonumber \\
&\leq & \sum_{s,e} p(e|\ACCd) \max_w \sum_{s',z,u} p(z,u|e,\ACCd) \left| p_w(s'|z,u,e, \ACCd) - \frac{1}{|S|} \right|  \nonumber \\
&\leq & |S| \sum_{e} p(e|\ACCd) \max_w \sum_{s,z,u} p_w(z,u|e,\ACCd) \left| p_w(s|z,u,e,\ACCd) - \frac{1}{|S|} \right|  \nonumber \\
& =& |S| \dwithw \stackrel{\text{\eqref{eq:dnow_dwithw}}}{=}  |S| \dnow.
\end{eqnarray}
\end{proof}
We therefore see that we can effectively work with the distribution $q(x,z,u,t,e)$ and $\dnow$. 
\begin{lemma}
\label{lem:Markov_acc}
For any probability distribution $q(x,z,u,t,e)$ satisfying (\ref{eq:q-cond1})-(\ref{eq:q-cond4}) it follows that
\be 
q(x|z,u,t,e,ACC) = q(x|z,u,ACC).
\ee
\end{lemma}  
\begin{proof}
For $(x,u) \notin ACC$ the claim holds trivially, since $q(x|z,u,t,e,ACC) = q(x|z,u,e,ACC) = 0$. For $(x,u) \in ACC$, we have
\begin{eqnarray}
q(x|z,u,t,e,ACC) 
&=& \frac{q(x,ACC_{u}|z,u,t,e)}{\sum_{(x,u) \in ACC} q(x,u|z,u,t,e)} \\ 
&=& \frac{q(x|z,u,t,e)}{\sum_{x \in ACC_u} q(x|z,u,t,e)} \nonumber \\ 
&\stackrel{\text{Eq. (\ref{eq:q-cond1})}}{=}&  \frac{q(x|z,u)}{\sum_{x \in ACC_u} q(x|z,u)} \nonumber \\ &=& q(x|z,u,ACC), \nonumber
\end{eqnarray}
which proves the claim.
\end{proof}

\begin{lemma}\label{lem:P(x|u)}
Consider the measure $q(x,z,u,t,e)$ satisfying conditions given by (\ref{eq:q-cond1})-(\ref{eq:q-cond4}). Let $\delta, \delta_{Az} > 0$ be constants and define the set
\begin{align}
\Azumadnc:=\{(z,u,e): Pr_{\sim q(x|z,u,e)}(\bar{L}\geq L+\delta_{Az})\leq\ep_{Az}\},
\end{align}
where $\ep_{Az}:=2e^{-\frac{1}{4}\delta_{Az}^2n}$ and
\ben 
&&L(x,u):=\frac1n \sum_{i=1}^n B_i({x}_{i}, u_{i}),\\
&&\overline{L}(x,z,u,e):=\frac1n\sum_{i=1}^n \mathbb E_{q(x_i,u_i|x_{<i},z,u_{<i},e)}B_i(x_i,u_i).
\een
Let $(z,u,e)\in \Azumadnc$. Then, for arbitrary $x\in\ACCd_u$, we obtain
\begin{align}
q(x|z,u,e)\leq\max\{\ep_{Az}, \gamma^{\mu n}\},
\end{align}
where
\begin{align}
\mu:=1-\sqrt{\delta+\delta_{Az}},\qquad 
\gamma = \frac13 \left(1+2\frac{\sqrt{\delta+\delta_{Az}}}{{(\frac12-\epsilon)^4}}\right).
\end{align}
\end{lemma}
\begin{proof}
Let $(z,u,e) \in \Azumadnc$ and $x \in \ACCd_u$.
We further define two sets: 
\begin{align}
\Xgood =\{x: |\Lazuma(x,u) - \overline{\Lazuma}(x,z,u,e)|\leq  \delta_{Az}\}
\end{align}
and 
\begin{align}
\Xbad=\left(\Xgood\right)^c.
\end{align}
Note that, for $(z,u,e)\in \Azumadnc$ and $x\in \Xbad$, we have 
\be
Pr_{\sim q(x|z,u,e)}(|\Lazuma(x,u) - \overline{\Lazuma}(x,z,u,e)|> \delta_{Az})\leq  \ep_{Az}.
\ee
Hence, for $x \in \Xbad$,
\be\label{estim_for_xbad}
q(x|z,u,e)\leq \ep_{Az}.
\ee
Let us now analyze the case, when $x\in \Xgood$. 
Since $x\in\ACCd_u$, which implies $(x,u)\in \ACCd$, we have $L(x,u)\leq\delta$. Further, for $x\in \Xgood$, 
we obtain
\begin{align}
\bar{L}(x,z,u,e)\leq\delta+\delta_{Az}.
\end{align}
Recall that
\begin{align}
\bar{B}_i=\mathbb E_{\sim q(x_i,u_i|x_{<i},z,u_{<i},e)}(B_i),
\end{align}
where $u_{<i}$ and $x_{<i}$ are the components of $u$ and $x$ respectively.
Following Lemma \ref{lem:linear_fraction}, we know that, for $\mu n$ positions $i$, where $\mu:=(1-\sqrt{\delta+\delta_{Az}})$, we get
\begin{align}
\bar{B}_i\leq\sqrt{\delta+\delta_{Az}}.
\end{align}
and $\bar{B}_i$ plays the role of $\bar{B}^{SV}$ here. 
Then, by Lemma \ref{lem:B_SV_rand}, for $\mu n$ positions $i$ there holds
\be
\label{eq:gamma}
q_{x_{<i},z,u_{<i},e}(x_i|u_i) \leq \gamma, 
\ee           
for any measurement outcome $x_i$ and setting  $u_i$.
Applying Lemma \ref{lem:min-entropy-0}, we get 
\be\label{estim_for_x_good}
q(x|z,u,e) \leq\gamma^{\mu n}.
\ee
Then, due to (\ref{estim_for_xbad}) and (\ref{estim_for_x_good}), the proof is completed.
\end{proof}

\begin{prop}\label{prop:max_x_small}
Given that the probability distribution $q(x,z,u,e)$ satisfies conditions (\ref{eq:q-cond1})-(\ref{eq:q-cond4}),
we have
\be
Pr_{\sim q(z,u,e|\ACCd)}\left(\max_x q(x|z,u,e,\ACCd) \leq \sqrt{\frac{\delta_1}{q(\ACCd)}}\,\,\right) \geq 1- \sqrt{\frac{\delta_1}{q(\ACCd)}}.
\ee
where 
\be
\delta_1= \gamma^{\mu n} + 2\ep_{Az}.
\ee
\end{prop}
\begin{proof}
Assume first that  $(z,u,e)\in \Azumadnc$. We then have 
\begin{eqnarray}
\label{eq:max_x_uyz}
\max_x q(x|z,u,e,\ACCd) &=& \max_x \frac{q(x,\ACCd|z,u,e)}{q(\ACCd|z,u,e)}  \\ 
&=& \frac{\max_{x\in\ACCd_u}q(x|z,u,e)}{q(\ACCd|z,u,e)} \nonumber \\ &\stackrel{\text{Lemma \ref{lem:P(x|u)}}}{\leq}& \frac{\max\{\ep_{Az}, \gamma^{\mu n}\}}{q(\ACCd|z,u,e)}.
\end{eqnarray}
We now consider 
\ben
&&\sum_{z,u,e}q(z,u,e|\ACCd) \max_x   q(x|z,u,e,\ACCd)  \nonumber \\
&&=\sum_{(z,u,e)\in\Azumad}q(z,u,e|\ACCd) \max_x   q(x|z,u,e,\ACCd) +\sum_{(z,u,e)\not\in\Azumad}q(z,u,e|\ACCd) \max_x   q(x|z,u,e,\ACCd).\nonumber\\
\een
We bound the first terms as follows:
\ben
\sum_{(z,u,e)\in\Azumadnc}q(z,u,e|\ACCd) \max_x   q(x|z,u,e,\ACCd) &\stackrel{\text{Eq. \eqref{eq:max_x_uyz}}}{\leq} & 
\sum_{(z,u,e)\in\Azumadnc}q(z,u,e|\ACCd) \frac{\max\{\ep_{Az}, \gamma^{\mu n}\}}{q(\ACCd|z,u,e)}  \nonumber \\
&\leq &\max\{\ep_{Az}, \gamma^{\mu n}\} \sum_{(z,u,e)} \frac{q(z,u,e)}{q(\ACCd)}  \nonumber \\ &=& \frac{\max\{\ep_{Az}, \gamma^{\mu n}\}}{q(\ACCd)}. \nonumber \\
\een	
\tgr{Let us now apply Lemma \ref{lemmaazuma}, taking $W_0=(e,z)$, $W_i=(x_i,u_i)$ for $i=1,\ldots, n$, and
$B_i$ given by \eqref{eq:Bell_indicator}. We obtain  $q(\Azumadnc) \geq 1 - \ep_{Az}$. Thus, the second term is bounded as follows}
\ben
\sum_{(z,u,e)\not\in\Azumadnc}q(z,u,e|\ACCd) \max_x   q(x|z,u,e,\ACCd) &\leq & \sum_{(z,u,e)\not\in\Azumadnc}q(z,u,e|\ACCd)  \nonumber \\
&\leq & \sum_{(z,u,e)\not\in\Azumadnc} \frac{q(z,u,e)}{q(\ACCd)} \nonumber \\
&\stackrel{\text{Lemma \ref{lemmaazuma}}}{\leq}& \frac{\ep_{Az}}{q(\ACCd)}.
\een
Altogether we find
\ben
\sum_{z,u,e}q(z,u,e|\ACCd) \max_x   q(x|z,u,e,\ACCd) &\leq& \frac{\max\{\ep_{Az}, \gamma^{\mu n}\}}{q(\ACCd)} +\frac{\ep_{Az}}{q(\ACCd)} \leq \frac{\gamma^{\mu n}+2\ep_{Az}}{q(\ACCd)}. \nonumber
\een
Applying Markov inequality and setting $\delta_1 = \gamma^{\mu n} + 2 \ep_{Az}$ completes the proof.
\end{proof}

The main theorem of this section is the following:

\begin{thm}
\label{thm:main_protocol_I}
Suppose we are given $\epsilon>0$. Set $\delta > 0$ such that  
\be
\label{eq:gamma_delta}
\frac13 \left(1+2\frac{\sqrt{2\delta}}{{(\frac12-\epsilon)^4}}\right) <1
\ee (see Fig. \ref{fig:delta-epsilon} trade-off between $\delta$ and $\epsilon$).
\tgr{Then for any probability distribution $p_w(x,z,u,t,e)$  
satisfying Eqs.  \eqref{eq:p-cond1}-\eqref{eq:p-cond6} 
there exists an extractor $s(x,t)$ with $|S| = 2^{\Omega(n^{1/4})}$ values, such that 
\be
\dcintro \cdot p(\ACCd) \leq 2^{-\Omega(n^{1/4})}, 
\ee
where $\dcintro$ is given by \eqref{eq:dcintro_I}.}
\end{thm}

\tgr{\begin{remark}\label{rem:p(ACC)=p_w,q(ACC)} Note, that due to first condition of \eqref{eq:p-cond1}, $p_w(\ACCd)$ 
does not depend on $w$, hence we could have written just $p(\ACCd)$ in the theorem. Moreover, we even have $q(\ACCd)=p(\ACCd)$.
\end{remark}}

\begin{figure}
\scalebox{0.60}{\includegraphics{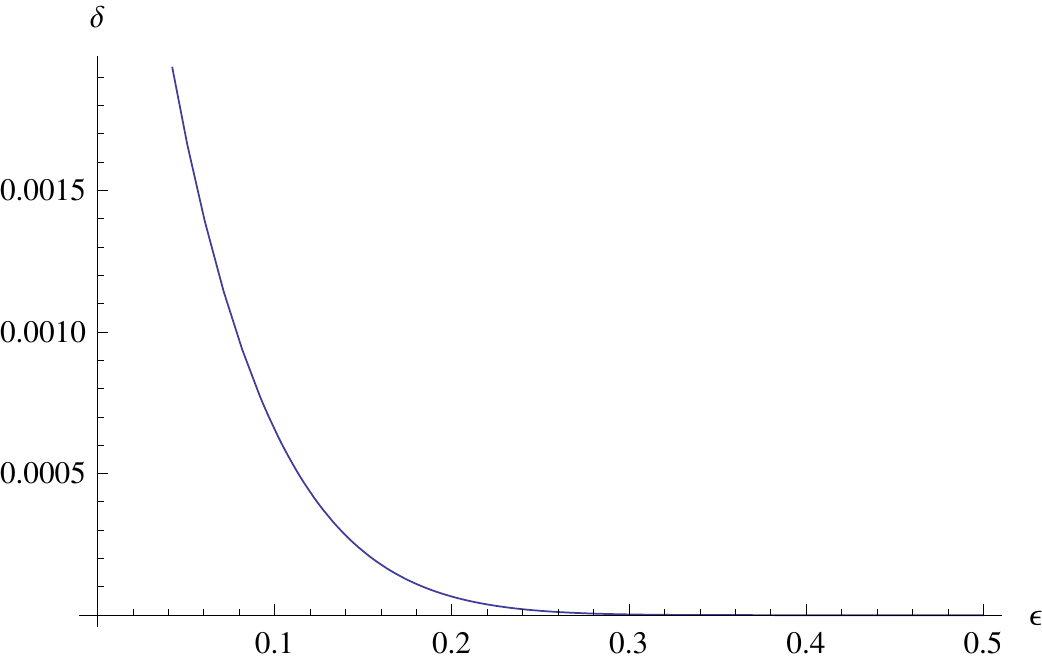}}
\caption{Trade-off between parameter $\epsilon$ of SV source and the amount of tolerated noise $\delta$}
\label{fig:delta-epsilon}
\end{figure}

\begin{proof}

Let $\delta>0$ satisfying Eq. (\ref{eq:gamma_delta}) be given. Set $\delta_{Az} = n^{-\frac14}$ so that $\ep_{Az} = 2e^{-\frac{1}{4}\delta_{Az}^2n} = 2^{-\Omega(\sqrt{n})}$. We  
consider only $n\geq n_0$ where $n_0$  is such that $\delta_{Az}\leq \delta$, i.e. $n_0=\lfloor\frac{1}{\delta^4}\rfloor$. 
Then $\mu\geq 1-\sqrt{2\delta}$ and $\gamma<1$. Now, let $\eta= \sqrt{\gamma^{\mu n}+2 \ep_{Az}}$,
so that $\eta=2^{-\Omega(\sqrt{n})}$.
We shall now consider distribution $q$ given by \eqref{eq:qvspw} and distance $\dnow$ of \eqref{eq:dist-no-w}. 
Suppose first that $q(\ACCd)\leq \eta$ and let us consider 
\be
\dnow = \sum_{z,u,e}q(z,u,e|\ACCd) \sum_{s=1}^{|S|}\biggl|q(s|z,u,e,\ACCd)-\frac{1}{|S|}\biggr|.
\ee 
Then, since  by definition $\dnow \leq 2$, 
we obtain $\dnow \cdot q(\ACCd)\leq 2^{-\Omega(\sqrt{n})}$. Now, suppose in turn, that  $q(\ACCd) \geq \eta$.
Then, from proposition \ref{prop:max_x_small} we get 
\be
\label{eq:good-prob}
Pr_{\sim q(z,u,e|\ACCd)} \left(\max_x q(x|z,u,e, \ACCd) \leq \sqrt{\eta} \right) \geq 1- \sqrt{\eta}.
\ee
Take the set $\good=\{(z,u,e): \max_x q(x|z,u,e, \ACCd) \leq \sqrt{\eta}\}$.
Then for $(z,u,e)\in \good$,  we have 
\be
H_{\min} \left( q(x|z,u,e,\ACCd)\right) \geq c \sqrt{n}.
\ee
We now consider the total probability distribution $q(x,z,u,t,e)$. 
From Eq. (\ref{eq:q-cond2}), we have 
that conditioned on $(z,u,e)$ the random variables $t$ and $x$ are independent 
\be
\label{x-t-ind}
q(x,t|z,u,e)=q(x|z,u,e) q(t|z,u,e).
\ee
Due to Lemma \ref{lem:Markov_acc}, we further obtain
\be
q(x,t|z,u,e, \ACCd)=q(x|z,u,e,\ACCd) q(t|z,u,e,\ACCd)
\ee
Moreover, by assumption (i.e. Eq. (\ref{eq:q-cond4})) $q(t|z,u,e)$ obeys the SV source conditions. Hence, if we show that $q(t|z,u,e,\ACCd) = q(t|z,u,e)$, then $q(t|z,u,e,\ACCd)$ obeys the SV source conditions as well. Note that
\be 
q(t|z,u,e,\ACCd) = \frac{\sum_{x \in \ACCd_u} q(t,x|z,u,e)}{\sum_{x \in \ACCd_u} q(x|z,u,e)} \stackrel{\text{Eq.(\ref{x-t-ind})}}{=} q(t|z,u,e).  
\ee
Therefore $q(t|z,u,e,\ACCd)$ is an min-entropy source and we have
\be
H_{\min}(q(t|z,u,e,\ACCd)) = c' n
\ee
where  $c'$ is a constant depending only on $\ep$. 
Thus $q(x,t|z,u,e, \ACCd)$ is a product of two min-entropy sources. 
By the application of extractor of lemma \ref{extractors} part (i) we obtain the output $s$ with 
\be
\sum_s \left|q(s|z,u,e,\ACCd) - \frac{1}{|S|}\right|\leq   2^{-\Omega(\sqrt{n})}.
\ee
For $(z,u,e)\not\in\good$, we use $ \sum_s \left|q(s|z,u,e,\ACCd) - \frac{1}{|S|}\right| \leq 2$ and obtain
\ben
\dnow 
&\leq & \sum_{(z,u,e)\in\good} q(z,u,e|\ACCd) 2^{-\Omega(\sqrt{n})} 
+\sum_{(z,u,e)\not\in\good} 2 q(z,u,e|\ACCd) \nonumber \\
&\leq & 2^{-\Omega(\sqrt{n})} +2 \frac{1-q(\good)}{q(\ACCd)}
\stackrel{\text{Eq. (\ref{eq:good-prob})}}{\leq} 2^{-\Omega(\sqrt{n})}+\frac{2 \sqrt{\eta}}{\eta}
=2^{-\Omega(\sqrt{n})}
\een
(recall we have set $\eta=2^{-\Omega(\sqrt{n})}$). 
We thus obtain that 
\be
\dnow \cdot q(\ACCd) \leq 2^{-\Omega(\sqrt{n})}
\ee
\tgr{By definition of $q$ (see also Rem. \ref{rem:p(ACC)=p_w,q(ACC)}), we have $q(\ACCd)=p(\ACCd)$, and  from Proposition \ref{lem:dc-dmax} 
we know that $\dcintro \leq |S| \dnow$.} So choosing $|S| = 2^{\Omega(n^{1/4})}$ the claim follows.  
\end{proof}

\section{Proof of correctness of Protocol II}
\label{sec:protocol2}
Protocol II considers the situation where the honest parties have two no-signaling devices, and share with the adversary Eve a no-signaling box $\{p(x,y^1, \ldots, y^N, z|u', v'^1, \ldots, v'^N, w)\}$. We work with the probability distribution $p(x,y^1, \ldots y^{N},z, u, v^1, \ldots, v^N, t, j, e, u',v'^1,\ldots v'^{N},w)$.
Here we consider $n$ uses of the first device with $u' (= u'_{1}, \dots, u'_{n})$ denoting the inputs for these, as well as $N$ blocks of the second device (each with $n$ uses) with $v'^1, \ldots, v'^{N}$ denoting the inputs for this device (note $v'^{k} = v'^{k}_{1}, \dots, v'^{k}_{n}$). The honest parties draw $u, v^1, \ldots, v^N$ as well as the bit strings $j$ and $t$ from the SV source. 
The parties input $u$ to the first device ($u' = u$) and $v^1,\ldots, v^{j}$ to the second device ($v'^k = v^k$ for $k \in [j]$). They obtain the corresponding outputs $x$ from the first device, and $y^1,\ldots, y^{j}$ from the second device. The remaining variables $y^{j+1},\ldots, y^N$ we define to be zero. The adversary Eve holds the bit string $e$ which is her classical information about the bits drawn from the SV source. Her input is denoted as $w$ with corresponding output $z$. To avoid cumbersome notation, we will use the shorthand $v' (v)$ to denote $v'^{1}, \dots, v'^{N}$ ($v^{1}, \dots, v^{N}$) as well as $y = y^1, \dots, y^N$ where there is no possibility of confusion. 
Now the initial box describing the initial variables and inputs is given by $p(x,y,z,u,v,t,j,e|u',v',w)$, a family of probability distributions labeled by $u', v', w$. The final output $s$ of the protocol the honest parties compute as a function of $x$, $y^j$ and $t$
\be
\label{eq:s_protocol2}
s=s(x,y^j,t).
\ee
Referring to notation in Sec. \ref{sec:assumptions},  for the  present protocol,  we have $\svin=(u,v)$, 
$\svhash=(t,j)$, $\sve=e$, $\about=(x,y)$, $\abin=(u',v')$, $\ein=w$, $\eout=z$.
The assumptions of Sec. III thus read as follows.

\begin{itemize}
\item{\bf No-signaling assumptions:}\\
We have full no-signaling between all parties and devices (see Fig. \ref{fig:protocol2}), i.e. 
\tgr{\ben 
\label{eq:NS_II1}
p(x,y|u',v',w) &=& p(x,y|u',v')\;\text{(no-signaling from Eve to honest parties)},  \\
\label{eq:NS_II2}
p(z|u',v',w) &=& p(z|w)\;\text{(no-signaling from honest parties to Eve)},  \\
\label{eq:NS_II3}
p(x,z|u',v',w)&=&p(x,z|u',w) \;
\\
\label{eq:NS_II4}
p(y,z|u',v',w)&=&p(y,z|v',w). 
\een
We also assume the following time ordered no-signaling conditions (see Fig. \ref{fig:protocol2}).
\ben
\label{eq:NS_IItons}
p(x_{k}|u',v',w,z,u,v,t,j,e) &=& p(x_{k} | u'_{\leq k},v',w,z,u,v,t,j,e) \; \; \; \forall k \in [n],\\
\label{eq:NS_IItons'}
p(y_k|u',v',w,z,u,v,t,j,e) &=& p(y_{k} | u',v'_{\leq k},w,z,u,v,t,j,e) \; \; \; \forall j \in [n].
\een }

\item{\bf Assumption A1-(II):} 
The devices do not signal to the SV source, i.e., the distribution of $(u,v,t,j,e)$ is independent of the inputs $(u', v', w)$:
\be
p(u,v,t,j,e |u', v', w) = p(u,v,t,j,e) \; \; \; \forall{(u,v,t,j,e, u', v', w)}.
\label{eq:assumption1-prot2}
\ee
\item{\bf Assumption A2-(II):} The form of the box is fixed independently of the SV source:
\be
p(x,y,z|u',v',w,u,v,t,j,e) = p(x,y, z|u',v',w) \; \; \; \forall{(x,y,z,u,v,t,j,e,u',v',w)}.
\label{eq:assumption2-prot2}
\ee

\item {\bf SV conditions:} 
\tgr{The distribution $p(u,v,t,j,e)$ satisfies an SV condition (\ref{SVdef}); in particular,
$p(t|u,v,j,e)$ satisfies Eq. (\ref{SVdef}) too.}
\end{itemize}
 
\tgr{\begin{remark}
Note that
\be\label{eq:NS_II3_more}
p(x|z,u',v',w)=\frac{p(x,z|u',v',w)}{p(z|u',v',w)}
\stackrel{Eq. (\ref{eq:NS_II3})}{=}\frac{p(x,z|u',w)}{p(z|u',v',w)}
\stackrel{Eqs. (\ref{eq:NS_II2}), (\ref{implication})}{=}\frac{p(x,z|u',w)}{p(z|u',w)}=p(x|z,u',w)\ee
and similary we may show that
\be\label{eq:NS_II4_more}
p(y|z,u',v',w)=p(y|z,v',w). 
\ee 
\end{remark}}

After the honest parties input $u' = u, v' = v$ in the protocol, we work with the distribution 
\be 
\label{eq:p_w-prot2}
p_w(x,y,z,u,v,t,j,e):=p(x,y,z,u,v,t,j,e|w)
\ee
Assumption A1-(II) ensures that this is a normalized probability distribution. 
The setup is shown in Figure \ref{fig:protocol2}.
\begin{figure}[!h]
\begin{center}
\includegraphics[trim=0cm 0cm 0cm 0cm, width=10cm]{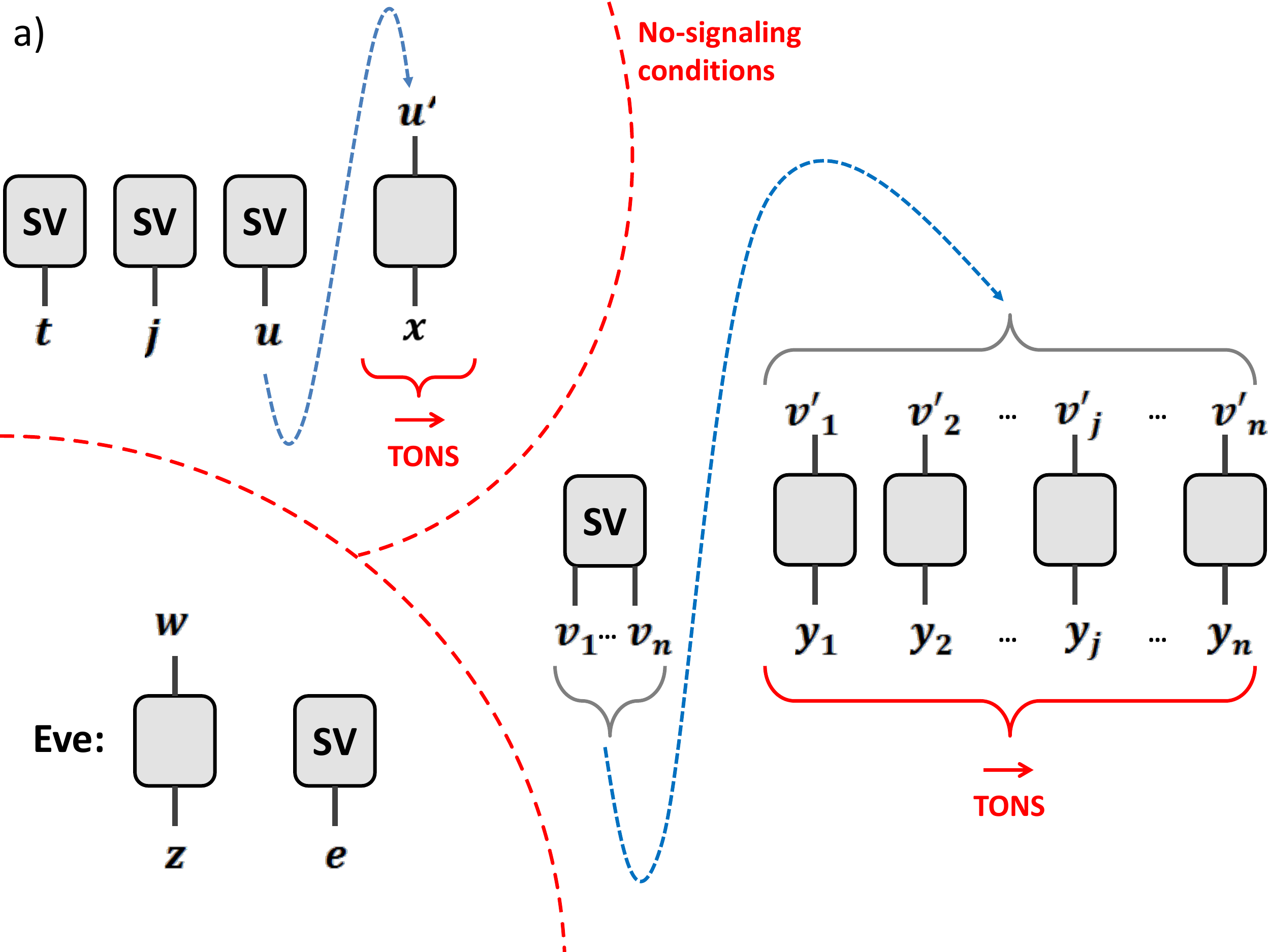}
\vspace{5mm}
\includegraphics[trim=0cm 10cm 0cm 0cm, width=10cm]{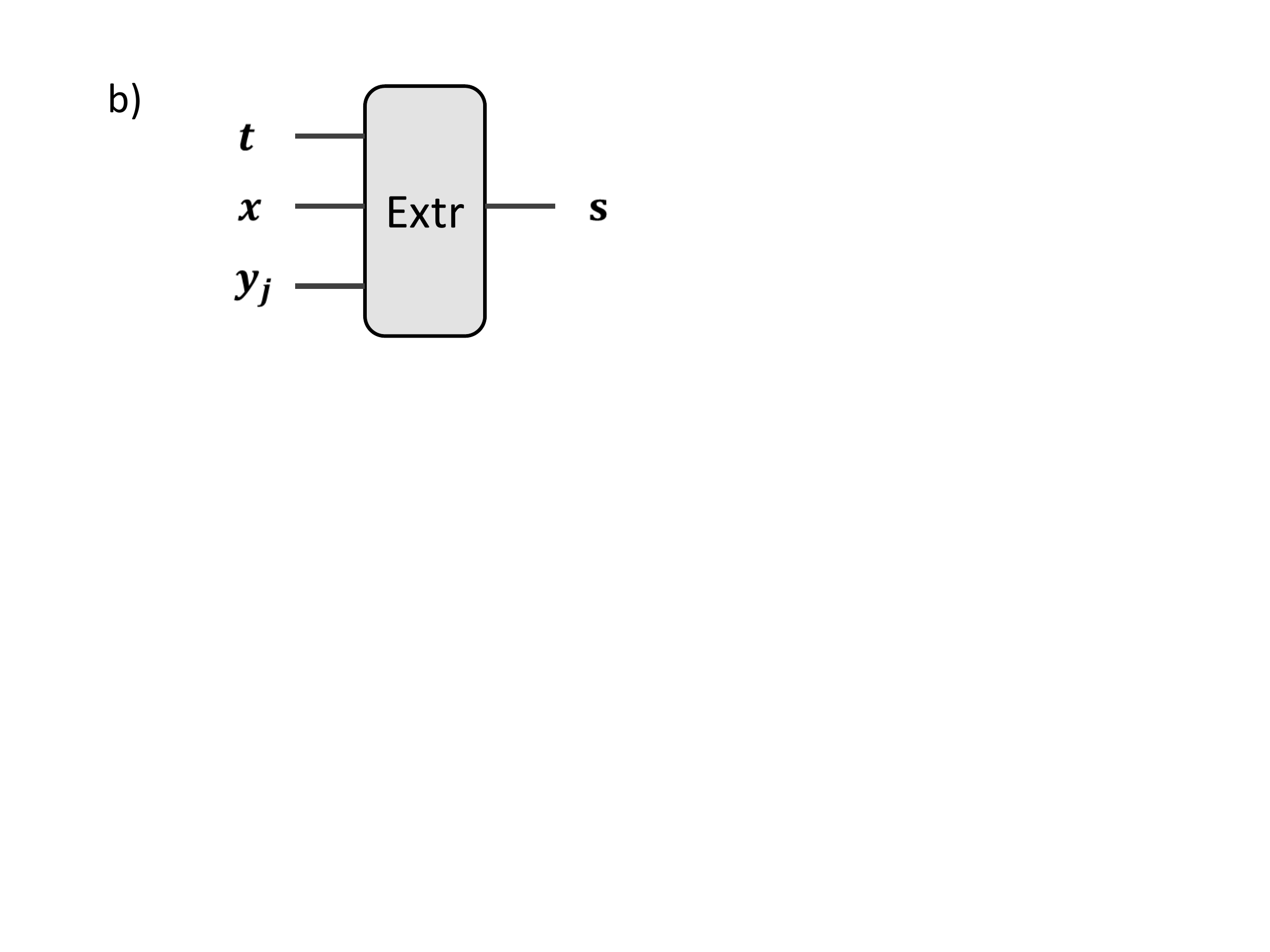}
\caption{Illustration of the protocol for randomness amplification from two devices (of four no-signaling parts each) 
 with one block of $n$ runs from the first device and $N$ blocks of $n$ runs from the second one.}
\label{fig:protocol2}
\end{center}
\end{figure}
We will use the further shorthand notation $R_j := ( y \setminus y^j, v \setminus v^j)$
and $M_j := (R_j, j)$.   
Further, we define the sets of acceptance as follows 
\begin{align} 
&\ACCd^1:=\{(x,u):L(x,u)\leq\delta\},\\ 
&\ACCd^j:=\{(y^j,v^j):L(y^j,v^j)\leq\delta\},\\
&\ACCdef:=\ACCd^1\cap\ACCd^j,
\end{align}
where $L(x,u) = \frac1n \sum_{i=1}^n B_i(x_i,u_i)$ and $L(y^j,v^j)=\frac1n\sum_{i=1}^n B_i(y^j_i, v^j_i)$.  
We also define the $u$-cut of the set $\ACCd^1$ as $\ACCd^1_{u} := \{x : (x, u) \in \ACCd^1\}$ and the $v^j$-cut of the set $\ACCd^j$ as $\ACCd^j_{v^j} := \{y^j: (y^j, v^j) \in \ACCd^j\}$. 

To quantify the quality of the output we will use the universally composable distance defined in  \eqref{eq:true_dc},
which is of the same form as the distance \eqref{eq:dcintro_I} used in  Protocol I 
\tgr{\begin{equation}
\label{eq:dcintroII}
\dcintroII :=  \sum_{s,e} \max_w \sum_{z} \left| p_w(s,z,e|\ACCdef) - \frac{1}{|S|} p_w(z,e|\ACCdef) \right|,
\end{equation}}
with $s$ given by \eqref{eq:s_protocol2}. The probability distribution $p_w(s,z,e|\ACCdef)$ is computed from the 
probability distributions $p_w(x,y,z,u,v,t,j,e)$ given by Eq. \eqref{eq:p_w-prot2}. Actually, in the proofs we will 
deal with slightly modified distance 

\begin{equation}
\dcoursII :=  \sum_{s,e} \max_w \sum_{z,u,v^j,M_j} \left| p_w(s,z,u,v^j,M_j,e|\ACCdef) - \frac{1}{|S|} p_w(z,u,v^j,M_j,e|\ACCdef) \right|,
\end{equation}
By triangle inequality, we have 
\be
\label{eq:dintro_doursII}
\dcintroII\leq \dcoursII
\ee
and hence it is enough to bound $\dcoursII$.

As in Protocol I, we now define the distance quantities $\dwithw$ and $\dnow$ to see that we can work with a~probability distribution without $w$. Let us first define the analogous quantity $\dwithwII$ for Protocol II
\tgr{\be
\label{eq:dist-w-prot2}
\dwithwII := \sum_{e} p(e|\ACCdef) \max_w \sum_{z,u,v^j, M_j} p_w(z,u,v^j,M_j|e,\ACCdef)  \sum_s \left| p_w(s|z,u,v^j,M_j,e,\ACCdef) - \frac{1}{|S|}\right|
\ee
where $p_w(s,z,u,v^j,M_j,e|\ACCdef)$ is computed from any family of probability distributions $\{p_w(x,y,z,u,v,t,j,e)\}$ satisfying 
\ben
\label{eq:p-cond-prot2_1}
&&p_w(x,y,u,v)=p(x,y,u,v) \quad \text{(by Eq. (\ref{eq:p-cond-prot2_1}) and A1-(II))}, \\
\label{eq:p-cond-prot2_2}
&& p_w (u,v,t,j,e)=p(u,v,t,j,e)  \quad \text{(follows from A1-(II))}, \\
\label{eq:p-cond-prot2_3} 
&& \forall_w\; p_w(x,y,z|u,v,t,j,e)= p_w(x,y,z|u,v)  \quad \text{(follows directly from A2-(II))}, \\
\label{eq:p-cond-prot2_4}
&& \forall_w\; p_w(x,y,z|u,v,t,j,e)= p_w(x,y,z|u,v,j,e)  \quad \text{(by A2-(II) and Rem. \ref{rem:implication})}, \\
\label{eq:p-cond-prot2_5}
&&p_w(x|z,u,v,t,j,e)= p_w(x|z,u,t,j,e)  \quad \text{(by A2-(II) and Eq. (\ref{eq:NS_II3_more}))}, \\
\label{eq:p-cond-prot2_6}
&&p_w(y|z,u,v,t,j,e)= p_w(y|z,v,t,j,e)  \quad \text{(by A2-(II) and Eq. (\ref{eq:NS_II4_more}))}, \\
&& p_w(x|z,u,v,t,j,e) = p_{w,z,v,t,j,e}(x|u)  \; \text{and} \; p_w(y|z,u,v,t,j,e) = p_{w,z,u,t,j,e}(y|v)  \nonumber\\
\label{eq:p-cond-prot2_7}
&&\qquad\qquad\qquad\qquad\qquad
\text{are time-ordered no-signaling (tons) boxes by (\ref{eq:NS_IItons}) and (\ref{eq:NS_IItons'})},\\
&& p_w(u|z,e),\: p_w(v|z,u,e),\:p_w(j|z,u,v,e) \; \text{and} \; p_w(t|z,u,v,j,e) \; \text{are SV sources (by A2-(II), } \nonumber\\
\label{eq:p-cond-prot2_8}
&& \text{Rem. (\ref{rem:implication}) and SV conditions - the proof goes in the same manner as in Rem. \ref{SV_proof})}.
\een}
For each $e$, let $w_e$ denote Eve's input $w$ and let $p_{w_{e}}(x,y,z,u,v,t,j|e)$ denote the corresponding distribution that achieves the maximum in Eq. (\ref{eq:dist-w-prot2}). Using the fact that $p_{w}(e) = p(e)$ and that $\ACCdef$ is a set of $(x,y,u,v)$ which obey $p_{w}(x,y,u,v) = p(x,y,u,v)$ (from Eq. (\ref{eq:p-cond-prot2_1})), we see that the distribution that achieves the maximum takes the form $p(e) p_{w_{e}}(x,y,z,u,v,t,j|e)$. 
We now set 
\be
\label{eq:qvspwII}
q(x,y,z,u,v,t,j,e) := p(e) p_{w_{e}}(x,y,z,u,v,t,j|e).
\ee
It can be readily seen that this $q(x,y,z,u,v,t,j,e)$ obeys the restrictions: 
\tgr{\ben
\label{eq:q-cond1-prot2}
&&q(x,y,z|u,v,t,j,e)= q(x,y,z|u,v),\\
\label{eq:q-cond2-prot2}
&&q(x,y,z|u,v,t,j,e)= q(x,y,z|u,v,j,e),\\
\label{eq:q-cond3-prot2}
&&q(x|z,u,v,t,j,e)=q(x|z,u,t,j,e) \; \text{and} \; q(y|z,u,v,t,j,e) = q(y|z,v,t,j,e),\\
\label{eq:q-cond4-prot2}
&&q(x|z,u,v,t,j,e)=q_{z,v,t,j,e}(x|u) \; \text{and} \; q(y|z,u,v,t,j,e) = q_{z,u,t,j,e}(y|v) \; \; \; \text{are tons boxes},\\
\label{eq:q-cond5-prot2}
&&q(u|z,e), q(v|z,u,e), q(j|z,u,v,e) \; \text{and} \; q(t|z,u,v,j,e) \; \; \; \text{are SV sources}.
\een}
We can therefore define 
\be
\label{eq:dist-no-w-prot2}
\dnowII :=  \sum_{z,u,v^j,M_j,e} q(z,u,v^j,M_j,e|\ACCdef) \sum_s \left| q(s|z,u,v^j, M_j,e,\ACCdef) - \frac{1}{|S|}\right|. 
\ee
and observe that $\dnowII = \dwithwII$. 

As for protocol I, the distance quantity $\dcintroII$ is bounded as in the following proposition.
\begin{prop}
\label{lem:dc-dmax-prot2}
For any distribution $p_w(x,y,z,u,v,t,j,e)$ given by (\ref{eq:p_w-prot2}) we have
\be
\dcintroII \leq |S| \dnowII.
\ee
\end{prop}
\begin{proof}
The proof is analogous to that of Lemma \ref{lem:dc-dmax} with the substitution $u \to (u, v^j M_j)$. 
\end{proof}
We therefore see that we can effectively work with the distribution $q(x,y,z,u,v,t,j,e)$ and $\dnowII$. 
We now define the set $\Adef$ for which  Azuma estimation works, as follows 
\ben
&&\Adefo = \{(z,u,v^j,M_j,e): Pr_{\sim q(x|z,u,v^j,M_j,e)}\left(\Lazumadef(x,z,u,M_j,e)\leq L(x,u)+\delta_{Az} \right) \geq 
1-\ep_{Az}\}, \nonumber\\
&&\Adefj = \{(z,u,v^j,M_j,e): Pr_{\sim q(y^j|z,u,v^j,M_j,e)}\left(\Lazumadef(y^j,z,v^j,M_j,e)\leq L(y^j,v^j)+\delta_{Az} \right) \geq 
1-\ep_{Az}\}, \nonumber\\
&&\Adef=\Adefo\cap\Adefj,
\een
where $\ep_{Az}:= 2e^{-\frac{1}{4}\delta_{Az}^2 n}$ and $\Lazumadef$ is of the form
\be 
\Lazumadef(x,z,u,M_j,e) 
:= \frac1n\sum_{i=1}^n \E_{\sim q(x_i,u_i|x_{<i},z,u_{<i},M_j,e)} B_i(x_i, u_i), 
\ee
\be
\Lazumadef(y^j,z,v^j,M_j,e) 
:= \frac1n\sum_{i=1}^n \E_{\sim q(y^j_i,v^j_i|z,y^j_{<i},v_{<i}^j,M_j,e)}B_i(y^j_i,v^j_i),
\ee
where $u_i$ and $v^j_{i}$ are distributed according to the measure $\nu$ from an SV source. 
Further, we define
\be
\label{eq:Sxidef}
\setadefnc= \{(z,u,v^j,M_j,e): Pr_{\sim q(x,y^j | z, u, v^j, M_j, e)} (\ACCdef) \geq \xi \}.
\ee
and
\be
\label{eq:Ddef}
\Ddef:=\{(z,u,v^j,M_j,e): \|q(x,y^j|z,u,v^j,M_j,e)- q(x| z,u,v^j,M_j,e)\otimes q(y^j|z,u,v^j,M_j,e)\|\leq \epdef \}.
\ee

We note that for any fixed $M_j$, the boxes corresponding to both blocks $q(x|z,u, M_j,e)$ and $q(y^j|z,v^j, M_j, e)$ are valid time ordered no-signaling boxes, as it was in the single device scenario.
Indeed, for the first quantity we have
\be
q(x|z,u,M_j,e)=q(x|z,y\setminus y^{j},u,v\setminus v^j,j,e)
= q(x|z,y^{<j},u,v\setminus v^j, j,e),
\ee
where we use $y^k=0$ for $k> j$ and in this distribution the inputs and outputs of the second device (that does not signal to the first) $v \setminus v^j$ and $y^{<j}$ are just labels for the distribution.  

Regarding the second device, we have   
\be
q(y^j|z,v^j,M_j,e)=q(y^j|y\setminus y^{j},z,v^j,v \setminus v^j, j, e)
=q(y^j| y^{<j},z,v^j, v\setminus v^j,j,e),
\ee
where we used the fact that $y^k = 0$ for $k > j$. Due to the assumption that, conditioned on the past, 
the box is still time-ordered no-signaling and because $v^{>j}$ are just random variables from SV source,
we have again that the latter distribution is a time-ordered no-signaling box 
and $v \setminus v^{j}$, $y^{<j}$ and $j,z ,e$ are just labels. 
We now observe the analogue of Lemma \ref{lem:P(x|u)}.
\begin{prop}
\label{prop:max_nc_def}
Consider the measure $q(x,y^j,z,u,v^j,M_j,e)$ satisfying conditions (\ref{eq:q-cond1-prot2})-(\ref{eq:q-cond5-prot2}). 
Let $\delta, \delta_{Az} > 0$ be constants and let $(z,u,v^j,M_j,e)\in A^{\delta_{Az}}$. 
Then, for arbitrary $x \in \ACCd^1_u$, we have
\be
q(x|z,u,M_j,e) \leq \max \{\gamma^{\mu n},\epazdef\},
\ee
and for arbitrary $y^j \in \ACCd^{j}_{v^j}$,
\be 
q(y^j|z,v^j,M_j,e) \leq \max \{\gamma^{\mu n},\epazdef \},
\ee
where
\begin{align}
\mu:=1-\sqrt{\delta+\delta_{Az}},\qquad 
\gamma = \frac13 \left(1+2\frac{\sqrt{\delta+\delta_{Az}}}{{(\frac12-\epsilon)^4}}\right).
\end{align}
\end{prop}
\begin{proof}
The proof is analogous to that of Lemma \ref{lem:P(x|u)} with the direct substitution $u \rightarrow (u, v^j, M_j)$
and noting that, \tgr{by no-signaling (i.e. Eq. (\ref{eq:q-cond3-prot2})) and Rem.  \ref{rem:implication}}, we have $q(y^j|z,u,v^j, M_j,e) = q(y^j|z,v^j, M_j,e)$ and $q(x|z,u,v^j,M_j,e) = q(x|z,u,M_j,e)$. 
\end{proof}

We now also have the analogue of Proposition \ref{prop:max_x_small}.
\begin{prop}
\label{prop:maxx_nc_deF}
Fix arbitrary $\delta, \delta_{Az}>0$ and consider the measure $q(x,y^j,z,u,v^j,M_j,e)$ that satisfies conditions (\ref{eq:q-cond1-prot2})-(\ref{eq:q-cond5-prot2}). We have 
\be
\label{upp-bound-x-nc}
Pr_{\sim q(z,u,v^j,M_j,e|\ACCd^1)}\left(\max_{x} q(x|z,u,v^j,M_j,e,\ACCd^1) 
\leq \sqrt{\frac{\delta_1}{q(\ACCd^1)}} \right) \geq 1- \sqrt{\frac{\delta_1}{q(\ACCd^1)}},
\ee
\be
\label{upp-bound-yj-nc}
Pr_{\sim q(z,u,v^j,M_j,e|\ACCd^j)}\left(\max_{y^j} q(y^j|z,u,v^j,M_j, e,\ACCd^j) 
\leq \sqrt{\frac{\delta_1}{q(\ACCd^j)}} \right) \geq 1- \sqrt{\frac{\delta_1}{q(\ACCd^j)}},
\ee
where 
\be
q(\ACCd^1)=Pr_{\sim q(x,z,u,v^j,M_j,e)} (\ACCd^1),\quad q(\ACCd^j)=Pr_{\sim q(y^j,z,u,v^j,M_j,e)} (\ACCd^j)
\ee
and 
\be
\delta_1 = \gamma^{\mu n} +2\ep_{Az}.
\ee
\end{prop}
\begin{proof}
The proof is analogous to that of Prop. \ref{prop:max_x_small} if we substitute $u \to (u,v^j M_j)$ and $\ACCd \to \ACCd^1$ for Eq.(\ref{upp-bound-x-nc}) and $\ACCd \to \ACCd^j$ for Eq.(\ref{upp-bound-yj-nc}).
\tgr{We also  use   Lemma \ref{lemmaazuma}  with  $W_i=(u_i,x_i)$ for $i\geq1$ and  $W_0= (z,v^j, M_j, e)$  for Eq.(\ref{upp-bound-x-nc}), and 
with  $W_i=(v^j_i,y^j_i)$ for $i\geq1$ and  $W_0= (z,u, M_j, e)$, for  Eq.(\ref{upp-bound-yj-nc}).}
\end{proof}

Here we shall prove a proposition that has no analogue in the security proof for Protocol I. 
It  says that, if the original distribution is close to product (due to our deFinetti-type result), 
then also the distribution conditioned upon acceptance will be close to product of distributions.
\begin{prop}
\label{prop:def_nc}
For arbitrary $\epdef > 0$,  $\xi\geq 2 \epdef$ and $M\equiv(z,u,v^j,M_j,e)\in \setadefnc$ suppose
\be
\label{eq:deFinetti-1}
\biggl\|q(x, y^j|M)-q(x|M)\ot q(y^j|M)\biggr\|\leq \epdef.
\ee 
Then we have 
\be
\biggl\|q(x,y^j|M,\ACCdef)-q(x|M,\ACCd^1)\ot q(y^j|M,\ACCd^j)\biggr\|\leq \frac{3 \epdef}{\xi^2}.
\label{eq:acc_prod}
\ee
\end{prop}

\begin{proof}
Since $\ACCdef=\ACCd^1 \cap \ACCd^j$, Eq. (\ref{eq:deFinetti-1})  implies 
\be
\left|q(\ACCdef|M)-q(\ACCd^1|M)q(\ACCd^j|M)\right|\leq \epdef.
\ee
Hence 
\ben
&&\biggl\|q(x, y^j|M,\ACCdef)-q(x|M,\ACCd^1) \ot q(y^j|M,\ACCd^j)\biggr\|   \nonumber \\
&&\leq\max_{\pm} \left\| \frac{q(x, y^j,\ACCdef|M)}{q(\ACCdef|M)} -\frac{q(x,\ACCd^1|M)\ot q(y^j,\ACCd^j|M)}{q(\ACCdef|M)(1\pm\kappa)}  \right\|
\een
where $\kappa= \frac{\epdef}{q(\ACCdef|M)}$. For $M \equiv(z,u,v^j,M_j,e)\in \setadefnc$, we have by definition of $\setadefnc$ in Eq. (\ref{eq:Sxidef}) that $q(\ACCdef|M) \geq \xi$, so that $\kappa \leq \frac{\epdef}{\xi} \leq \frac{1}{2}$ for $\xi \geq 2 \epdef$.  
Sandwiching the above with $\frac{q(x,\ACCd^1|M)\ot q(y^j,\ACCd^j|M)}{q(\ACCdef|M)}$ and using triangle inequality we get 
\ben
\label{eq:ACC_def}
&&\biggl\|q(x, y^j|M,\ACCdef)-q(x|M,\ACCd^1)\ot q(y^j|M,\ACCd^j)\biggr\| \nonumber \\
&&\leq  \left\| \frac{q(x, y^j,\ACCdef|M)}{q(\ACCdef|M)} -\frac{q(x,\ACCd^1|M)\ot q(y^j,\ACCd^j|M)}{q(\ACCdef|M)} \right\| \nonumber \\
&& +\left\|\frac{q(x,\ACCd^1|M)\ot q(y^j,\ACCd^j|M)}{q(\ACCdef|M)} \biggl( 1 -\frac{1}{1\pm\kappa}\biggr) \right\|
\een
Now we can bound the first term on the right hand side of Eq. (\ref{eq:ACC_def}) as follows
\ben
&&\left\| \frac{q(x,y^j,\ACCdef|M)}{q(\ACCdef|M)} -\frac{q(x,\ACCd^1|M)\ot q(y^j,\ACCd^j|M)}{q(\ACCdef|M)} \right\|  \nonumber \\
&&\leq\frac{1}{q(\ACCdef|M)} \left\|	 q(x,y^j|M) - q(x|M)\ot q(y^j|M) \right\| \stackrel{\text{Eq. (\ref{eq:deFinetti-1})}}{\leq} \frac{\epdef}{q(\ACCdef|M)},
\een
since projecting onto $\ACCdef$ is a trace non-increasing channel and hence cannot increase trace norm.  
The second term on the right hand side of Eq. (\ref{eq:ACC_def}) is simply bounded as 
\be
\left\|\frac{q(x,\ACCd^1|M)\ot q(y^j,\ACCd^j|M)}{q(\ACCdef|M)} \biggl( 1 -\frac{1}{1\pm\kappa}\biggr) \right\|
\leq \frac{1}{q(\ACCdef|M)}\left| \frac{\kappa}{1-	 \kappa}\right| \leq \frac{2\kappa}{q(\ACCdef|M)}, 
\ee
where $\kappa \leq \frac{\epdef}{\xi} \leq \frac{1}{2}$. 
Inserting these into \eqref{eq:ACC_def} and using $1 \geq q(\ACCdef|M) \geq \xi$ for $M \in \setadefnc$, 
we get
\be
\biggl\|q(x,y^j|M,\ACCdef)-q(x|M,\ACCd^1)\ot q(y^j|M,\ACCd^j)\biggr\|\leq 3 \frac{\epdef}{\xi^2}.
\ee
\end{proof}

Now we will combine Prop.  \ref{prop:maxx_nc_deF} which gives that from device 1 and device 2 
we obtain probability distributions with upper bounded maximal probability and Prop. 
\ref{prop:def_nc} which says that the joint probability distribution is close to a product distribution.

\begin{thm}
\label{thm:main_protocol_II}
Suppose we are given $\epsilon>0$. Set $\delta > 0$ such that  
\be
\label{eq:gamma_delta_2}
\frac13 \left(1+2\frac{\sqrt{2\delta}}{{(\frac12-\epsilon)^4}}\right) <1
\ee (see Fig. \ref{fig:delta-epsilon} trade-off between $\delta$ and $\epsilon$).
\tgr{Then for arbitrary family of  probability distributions $p_w(x,y,z,u,v,t,j,e)$ satisfying conditions 
(\ref{eq:p-cond-prot2_1})-(\ref{eq:p-cond-prot2_8})
there exists an extractor $s(x,y^j,t)$ with $|S| = 2^{n^{c}}$ values for a constant $c>0$ (depending on $\epsilon$) such that 
\be
\dcintroII \cdot p(\ACCdef)  \leq 2^{-n^{\Omega(1)}}. 
\ee
where $\dcintroII$ is given by \eqref{eq:dcintroII}.}
\end{thm}
\tgr{\begin{remark}\label{remII} Note, that due to \eqref{eq:p-cond-prot2_1}, $p_w(\ACCdef)$ 
does not depend on $w$, hence we wrote just $p(\ACCdef)$ in the theorem. Moreover, we even have $q(\ACCdef)=p(\ACCdef)$.\end{remark}}

\begin{proof}
Set $M = (z,u, v^j,M_j, e)$ and the probability distribution $q$ given by \eqref{eq:qvspwII}. Let us consider $\dnowII$, as in Eq. (\ref{eq:dist-no-w-prot2}), 
\ben
\label{eq:ddef} 
\dnowII &=&  \sum_{M} q(M|\ACCdef) \sum_s \left| q(s|M, \ACCdef) -\frac{1}{|S|}\right| \nonumber \\
&=& \sum_{M \not\in \Ddef} q(M| \ACCdef) \sum_s \left| q(s|M, \ACCdef) - \frac{1}{|S|}\right| +  \sum_{M \in \setadefnc \cap \Ddef } q(M| \ACCdef) \sum_s \left| q(s|M, \ACCdef) - \frac{1}{|S|}\right| \nonumber \\ &&+ \sum_{M \in (\setadefnc)^c \cap \Ddef} q(M| \ACCdef) \sum_s \left| q(s|M, \ACCdef) - \frac{1}{|S|}\right| \nonumber \\
\een
We consider the three terms separately. Let us consider the first term with $M \not \in \Ddef$. By Eq. (\ref{deFinetti}) from Lemma \ref{lem:deFinetti-bound1}, we know that 
\be 
\mathbb{E}_{M \sim q(M)} \|q(x,y^j|M)- q(x|M)\otimes q(y^j|M)\| \leq \epdef^2,
\ee
where $\epdef^2=c \sqrt{n} N^{\frac12\log(\frac12+\ep)}$ and $c$ is an absolute constant. 
Applying Markov inequality to the above, we get that 
\be 
\Pr_{q(M)} \left( \|q(x,y^j|M)- q(x|M)\otimes q(y^j|M)\| \leq \epdef \right) \geq 1 - \epdef, 
\ee 
so that $q(\Ddef) \geq 1 - \epdef$. We therefore have, using $\sum_s \left| q(s|M, \ACCdef) - \frac{1}{|S|}\right| \leq 2$, that
\ben\label{eq:Mz_not_in_D}
\sum_{M \not\in \Ddef} q(M| \ACCdef) \sum_s \left| q(s|M, \ACCdef) - \frac{1}{|S|}\right| &\leq & 2 \sum_{M \not\in \Ddef} \frac{q(M, \ACCdef)}{q(\ACCdef)} \nonumber \\
& \leq & 2 \sum_{M \not \in \Ddef} \frac{q(M)}{q(\ACCdef)} \leq \frac{2 \epdef}{q(\ACCdef)}.
\een
Consider now the second term with $M \in \setadefnc \cap \Ddef$. 
Proposition \ref{prop:maxx_nc_deF}  implies  that with probability 
\be
1 - \sqrt{\frac{\gamma^{\mu n}+2\ep_{Az}}{q(\ACCd_1)}} - \sqrt{\frac{\gamma^{\mu n}+2\ep_{Az}}{q(\ACCd_2)}}
\geq 1 - 2 \sqrt{\frac{\gamma^{\mu n}+2\ep_{Az}}{q(\ACCdef)}}
\ee
we have that for $h := \frac12 \left[\log (\gamma^{\mu n} + 2 \ep_{Az}) + \log q(\ACCdef) \right]$,
\be
H_{\min}(q(x|M,\ACCd^1)) \geq h \; \; \text{and} \; \;
H_{\min}(q(y^j|M,\ACCd^j)) \geq h. 
\ee
Moreover, by condition (\ref{eq:q-cond4-prot2}) and (\ref{eq:q-cond2-prot2}) (i.e. $q(t|x,y,z,u,v,j,e) = q(t|z,u,v,j,e)$), we obtain that $q(t|M, \ACC)$ also satisfies the SV source conditions, i.e., $H_{\min}(q(t|M, \ACC)) \geq c n$ for some constant $c$ depending on $\varepsilon$. 
By Lemma \ref{extractors} part (ii), there exists an extractor that extracts $\Theta(h)$ bits $s = s(x, y^j,t)$. Let us denote the action of extractor by $\Extr$.
The extractor acts on 
\be\label{eq:q|M}
q(x, y^j,t| M, \ACCdef) = q(t|M, \ACCdef) \ot q(x, y^j |M, \ACCdef)
\ee 
producing the output $s$ with distribution $q(s|M, \ACCdef)$ and so
\be 
\Extr \left(q(x,y^j,t|M,\ACCdef) \right) \equiv q(s|M,\ACCdef).
\ee
The action of the extractor on an ideal product distribution produces an output with some distribution which we denote $\bar{q}(s|M)$
\be\label{bar_q}
\Extr \left(q(t|M, \ACCdef) \ot q(x|M,\ACCdef^1) \ot q(y^j|M,\ACCdef^j) \right) \equiv \bar{q}(s|M)
\ee
From Lemma \ref{extractors} part (ii), we know that 
\be
\label{eq:q_s}
\sum_s  \left | \bar{q}(s|M) - \frac{1}{|S|}\right| \leq 2^{-h^{\Omega(1)}},
\ee
with $|S|=2^{\Theta(h)}$. Since $M\in \Ddef$, by Prop. \ref{prop:def_nc} we have  that for $M \in  \setadefnc$ 
\be
\|q(x, y^j|M,\ACCdef)- q(x|M,\ACCd^1)\ot q(y^j|M,\ACCd^j)\|\leq \frac{3\epdef}{\xi^2},
\ee
so that, using Eq. (\ref{eq:q|M}), we have
\be 
\| q(t, x, y^j|M, \ACCdef) - q(t| M, \ACCdef) \ot q(x|M, \ACCd^1) \ot q(y^j | M, \ACCd^j) \| \leq \frac{3\epdef}{\xi^2}.
\ee
Hence, according to monotonicity of trace norm and Eq. (\ref{bar_q}) we obtain
\be
\label{eq:q_def}
\|q(s|M,\ACCdef)-\bar{q}(s|M) \| \leq \frac{3\epdef}{\xi^2}
\ee
Thus, by using \eqref{eq:q_s} and \eqref{eq:q_def} and applying triangle inequality we get 
\be
\sum_s \left| q(s|M,\ACCdef) - \frac{1}{|S|} \right|\leq 2^{-h^{\Omega(1)}} +  \frac{3\epdef}{\xi^2}
\ee
and therefore the second term in Eq. (\ref{eq:ddef}) is estimated as follows
\be
\label{eq:Mz_in_Seta}
\sum_{M \in\setadefnc \cap \Ddef} q(M|\ACCdef) \sum_s \left | q(s|M, \ACCdef) - \frac{1}{|S|} \right| \leq 2^{-h^{\Omega(1)}} +  \frac{3\epdef}{\xi^2}.
\ee
Now, we deal with the third term, i.e. with $M \in (\setadefnc)^c \cap \Ddef$. For $M \in (\setadefnc)^c$ we have that $q(\ACCdef|M) \leq \xi$. Using $\sum_s \left | q(s|M, \ACCdef) - \frac{1}{|S|} \right| \leq 2$, we obtain
\ben
\label{eq:Mz_not_seta}
\sum_{M \in (\setadefnc)^c \cap \Ddef} q(M|\ACCdef) \sum_s\left | q(s|M\ACCdef) - \frac{1}{|S|}\right| &\leq & 
2\sum_{M \in(\setadefnc)^c \cap \Ddef} q(M|\ACCdef)  \nonumber \\
&\leq & 2  \sum_{M \in (\setadefnc)^c \cap \Ddef} \frac{q(M) q(\ACCdef|M)}{q(\ACCdef)} \leq  \frac{2\xi}{q(\ACCdef)} \nonumber \\
\een
Inserting  \eqref{eq:Mz_not_in_D}, \eqref{eq:Mz_not_seta} and \eqref{eq:Mz_in_Seta} into  \eqref{eq:ddef}, we obtain 
\be
\label{eq:ddef_bound}
\dnowII \leq \frac{2 \epdef}{q(\ACCdef)} + 2^{-h^{\Omega(1)}} + \frac{3\epdef}{\xi^2} +  \frac{2\xi}{q(\ACCdef)},
\ee
where $\xi$ is an arbitrary number satisfying $\xi \geq 2 \epdef$ (as in proposition \ref{prop:def_nc}),
and $h=\frac12 \left[ \log (\gamma^{\mu n} + 2 \ep_{Az}) + \log q(\ACCdef) \right]$.
We now analyze the product $\dnowII \cdot p(\ACCdef)$. 

As in the case of Protocol I, we set $\delta_{Az} = n^{-\frac14}$ so that 
$\ep_{Az} = 2e^{-\frac{1}{4}\delta_e^2n} = 2^{-\Omega(\sqrt{n})}$. We consider only $n\geq n_0$ where $n_0$ 
is such that $\delta_e\leq \delta$ (i.e. $n_0=\lfloor\frac{1}{\delta^4}\rfloor$). Then $\mu\geq 1-\sqrt{2\delta}$ and $\gamma<1$.
Now, let $\eta=\sqrt{\gamma^{\mu n}+2 \ep_{Az}}$,
so that $\eta=2^{-\Omega(\sqrt{n})}$. Suppose first that $q(\ACCd)\leq \eta$ then since $\dnowII \leq 2$, 
we have $\dnowII \cdot q(\ACCdef)\leq 2 \eta$. Suppose now, that $q(\ACCdef)\geq \eta$. 
Then setting $\xi = \eta^2$, we get  $2^h= 2^{\Omega(\sqrt{n})}$ and, using \eqref{eq:ddef_bound}, we obtain
\be
\dnowII \leq 2^{-n^{\Omega(1)}} + \frac{5\epdef}{\eta^4} + 2 \eta.
\ee
Finally, recall that, due to Lemma \ref{lem:deFinetti-bound1}, 
$\epdef^2 = c \sqrt{n} N^{\frac12\log(\frac12+\ep)}$ where $c$ is absolute constant. Setting the number of blocks $N=2^{n}$ so that $\epdef^2 = 2^{-\Omega(n)}$,
we obtain 
\be
\dnowII \leq 2^{-n^{\Omega(1)}}.
\ee
so that 
\be
\dnowII \cdot q(\ACCdef)\leq 2^{-n^{\Omega(1)}}.
\ee
\tgr{By definition of $q$ and Rem. \ref{remII}, we have $q(\ACCd)=p(\ACCd)$ and from Prop. \ref{lem:dc-dmax-prot2} we know that $\dcintroII \leq |S| \dnowII$.
We then see that there exists some constant $c>0$ such that setting $|S| = 2^{n^{c}}$, 
we get $\dcintroII \cdot p(\ACCdef)\leq 2^{-n^{\Omega(1)}}$.} This completes the proof.
\end{proof}

\section{Conclusion and Open Questions} 

We presented a protocol for obtaining secure random bits from an arbitrarily (but not fully deterministic) Santha-Vazirani source. The protocol uses a finite number (as few as four for the Bell inequality considered here) of no-signaling devices, and works even with correlations attainable by noisy quantum mechanical resources. Moreover the correctness of the protocol is not based on quantum mechanics and only requires the no-signaling principle.

We leave the following open questions to future research:

\begin{itemize}

\item Is there an efficient protocol for device-independent randomness amplification with a constant
number of devices, and tolerating a constant rate of noise, whose correctnesses only
assume limited independence between the source and the devices?

\item Does Protocol I and II work even for a public SV source (in which the bits are drawn from the source and communicated to all the parties)?

\item Is there a bipartite Bell inequality with the property that it is algebraically violated by quantum correlations, and for all settings in the Bell expression, the probabilities of the outputs are bounded away from one? Our protocol could be applied with such a Bell expression with a significant reduction in the number of no-signaling devices required. 

\item Is there a protocol that can tolerate a higher level of noise? What if we assume the validity of quantum mechanics?

\item Can we amplify randomness with only a finite number of devices from other different types of sources? A particularly interesting case is the min-entropy source \cite{CSW14, Scarani, Plesch}. 

\item A more technical question is to improve the de Finetti theorem given in \cite{Brandao, Brandao2}. What are the limits of de Finetti type results when the subsystems are selected from a Santha-Vazirani source?

\item Finally suppose one would like to realize device-independent quantum key distribution with only an imperfect SV source as the randomness source. Is there an efficient protocol for this task tolerating a constant rate of noise and gives a constant rate of key? 
Here the question is open for both quantum-mechanical and no-signaling adversaries. 

\end{itemize}

\section{Acknowledgments}

We thank Rotem Arnon-Friedman for discussions and for pointing out an error in an earlier version of this paper. The paper is supported by ERC AdG grant QOLAPS, EC grant RAQUEL and by Foundation for Polish Science TEAM project co-financed by the EU European Regional Development Fund. FB acknowledges support from EPSRC and Polish Ministry of Science and Higher Education Grant no. IdP2011 000361. Part of this work was done in National Quantum Information Center of Gda\'{n}sk. Part of this work was done when FB, RR, KH and MH attended the program ``Mathematical Challenges in Quantum Information" at the Isaac Newton Institute for Mathematical Sciences in the University of Cambridge. Another part was done in the programme "Quantum Hamiltonian Complexity" in the Simons Institute foe the Theory of Computing. Finally MH thanks Department of Physics and Astronomy and Department of Computer Science of UCL, where part of this work was also performed, for hospitality.

\bibliographystyle{apsrev}

\end{document}